\documentclass[12pt,a4paper]{article}  

 \usepackage[skins,theorems]{tcolorbox}

\newtcolorbox{whitebox}{colback=white,colframe=black,boxrule=0.5mm,arc=4mm,auto outer arc}

\tcbset{highlight math style={enhanced,
  colframe=red,colback=white,arc=0pt,boxrule=1pt}}
  \usepackage[bookmarksopen, bookmarksnumbered, bookmarksopenlevel=2]{hyperref}
  \usepackage{tikz}
  \usepackage{tikz-3dplot}
  \usepackage{multirow}
 \usetikzlibrary{calc}
 \usetikzlibrary{decorations} %
 \usepackage[UKenglish]{babel}
 \usepackage[toc,page]{appendix}
 \usepackage{amsmath}
 \usepackage{amssymb}
 \usepackage{amsthm}
 \usepackage{graphicx}
 \usepackage{hhline}
 \usepackage[bf]{caption}
\usepackage{cite}
\usepackage[vcentermath]{youngtab}
\usepackage{geometry}
\usepackage{slashed}
\usepackage{color}
\usepackage{stackrel}
\usepackage{tikz-cd} 
\usepackage{tkz-euclide}
\usepackage{cancel} 
\usepackage{subcaption}
\usepackage[normalem]{ulem}
\usepackage{mdframed}
\usepackage{adjustbox}
\usepackage{tcolorbox}
\usepackage{enumitem}

\usepackage{empheq}
\usepackage{arydshln}
\usepackage[dvipsnames]{xcolor}
\newenvironment{eqn*}{\begin{equation*}\begin{aligned}}{\end{aligned}\end{equation*}\noindent}

\newtheorem{claim}{Claim}

\newcommand{\bqa}{\begin{eqnarray}}
\newcommand{\eqa}{\end{eqnarray}}

\hypersetup{
    pdftitle={},
    pdfauthor={},
    pdfsubject={}
}
\numberwithin{equation}{section}
\numberwithin{table}{section}

\makeatletter

\definecolor{BF}{HTML}{f903d7}

\newtheorem{definition}{Definition}

\newtheorem{example}{Example}

\newtheorem{swamp}{Quantum Gravity Principle}

\newcommand{\be}{\begin{equation}}
\newcommand{\ee}{\end{equation}}
\newcommand{\beq}{\begin{equation}}
\newcommand{\eeq}{\end{equation}}
\newcommand{\ba}{\begin{aligned}}
\newcommand{\ea}{\end{aligned}}

\newcommand{\bea}{\begin{eqnarray}}
\newcommand{\eea}{\end{eqnarray}}

\newcommand{\cC}{\mathcal{C}}

\newcommand{\cN}{\mathcal{N}}

\newcommand{\cA}{\mathcal{A}}

\newcommand{\cF}{\mathcal{F}}

\newcommand{\cR}{\mathcal{R}}
\newcommand{\cS}{\mathcal{S}}

\newcommand{\cM}{\mathcal M}
\newcommand{\cQ}{\mathcal Q}

\newcommand\bi{\begin{itemize}}
\newcommand\ei{\end{itemize}}

\def\Im{\mathop{\mathrm{Im}}\nolimits}

\def\unit{{1\kern-.65ex {\rm l}}}
\def\1{{1\kern-.65ex {\rm l}}}

\def\ii{{\rm i}}

\tcbset {
	base/.style={
		arc=0mm, 
		bottomtitle=0.5mm,
		boxrule=0mm,
		colbacktitle=black!10!white, 
		coltitle=black, 
		fonttitle=\bfseries, 
		left=2.5mm,
		leftrule=1mm,
		right=3.5mm,
		title={#1},
		toptitle=0.75mm,
		breakable
	}
}

\newtcolorbox{subbox}[1]{
	colframe=black!30!white,
	base={#1}
}

\newcount\hour \newcount\minute
\hour=\time \divide \hour by 60
\minute=\time
\def\now{%
\ifnum \hour<13
  \ifnum \hour=0 \advance \hour by 12 \number\hour:\else \number\hour:\fi%
     \ifnum \minute<10 0\fi%
     \number\minute%
\ A.M.%
\else \advance \hour by -12 \number\hour:%
  \ifnum \minute<10 0\fi%
  \number\minute%
  \ P.M.%
\fi%
}

\makeatother

\begin{document}
\begin{flushright}
{\tt\normalsize ZMP-HH/26-32}\\
\end{flushright}

\vskip 40 pt
\begin{center}
{\large \bf
Quantum Gravity Principles for 4d $\cN=2$ Supergravity Theories\\ \vspace{0.35cm} and Asymptotic Hodge Theory
} 

\vskip 11 mm

Lukas Kaufmann 
and Timo Weigand 

\vskip 11 mm
\small \textit{II. Institut f\"ur Theoretische Physik, Universit\"at Hamburg, Notkestrasse 9,\\ 22607 Hamburg, Germany} \\[3 mm]
\small \textit{Zentrum f\"ur Mathematische Physik, Universit\"at Hamburg, Bundesstrasse 55, \\ 20146 Hamburg, Germany  }   \\[3 mm]

\end{center}

\vskip 7mm

\begin{abstract}

We study the behaviour of general four-dimensional $\cN=2$ supergravity theories near boundaries at infinity of the vector multiplet moduli space, with a focus on a clear distinction between properties that follow from supergravity alone rather than from a consistent UV completion in quantum gravity.
The basis for our analysis is the appearance of a real variation of polarised Hodge structure governing the vector multiplet moduli space of $\cN=2$ supergravity,  combined with very recent advances in the mathematics literature.
The resulting classification of infinite distance degenerations exactly agrees, at the level of gauge couplings, with a possible interpretation of the limits as decompactification limits to five or six dimensions or as weakly coupled string limits, as expected from a quantum gravity perspective.
This observation extends previous results found in the context of string compactifications on Calabi--Yau varieties and, remarkably, holds for general four-dimensional $\cN=2$ supergravity theories regardless of a UV completion. We furthermore show the existence of a K\"ahler basis in which the triple couplings of the prepotential are non-negative, irrespective of a geometric origin of the  supergravity theory.
Finally we propose two non-trivial constraints for a four-dimensional $\cN=2$  supergravity to be compatible with quantum gravity: Integrality of the Hodge structure as a necessary condition for the completeness hypothesis and the existence of a simple normal crossing compactification of the moduli space to ensure compatibility of the  weakly coupled asymptotic gravitational duality frames. 

\end{abstract}

\vfill

\thispagestyle{empty}
\setcounter{page}{0}
 \newpage

\tableofcontents
\vspace{25pt} 

\setcounter{page}{1}

\section{Introduction}
Candidate universal principles of quantum gravity \cite{Vafa:2005ui} can be investigated from two complementary directions. Bottom-up approaches take a low-energy perspective
and infer consistency constraints from, for example, semi-classical black holes and general properties of gravitational effective field theories. Top-down approaches test these proposals in explicit, UV-complete theories. In string theory, many quantum gravity conjectures have been verified quantitatively across broad classes of compactifications. For reviews and an overview of the literature, see for example \cite{Palti:2019pca,Grana:2021zvf,vanBeest:2021lhn,Agmon:2022thq}.

In this context, the systematic study of moduli spaces of supergravity theories from a quantum gravity perspective forms an interesting middle ground: While not necessarily tied to concrete UV completions, it is sufficiently restrictive to allow for a quantitative and meaningful analysis.
A central question is which properties, in addition to the constraints from supersymmetry alone, a supergravity moduli space  must have in order to originate from a consistent quantum gravity theory, or, more specifically, from a compactification of string theory.
For example, vector multiplet moduli spaces in 4d ${\cal N}=2$ supergravity are projective special K\"ahler manifolds~\cite{deWit:1984wbb,Castellani:1990tp,Castellani:1990zd,Strominger:1990pd,Craps:1997gp}, and a long-standing challenge is to single out the precise subset of these which form (part of) a Calabi--Yau moduli space.

Identifying such additional constraints on supergravity moduli spaces is also important in view of  a potential drawback of the top-down justification of quantum gravity principles as currently performed in the literature:
 With few exceptions, e.g. \cite{Basile:2022zee,Baines:2025upi,Baines:2026aug}, the focus has been on situations with (extended) supersymmetry at low energies. This raises the important question to what extent the proposed  principles identify genuine features of quantum gravity, rather than reflecting our focus on the supersymmetric lamppost.

In this work, we address these questions in the context of the Swampland Distance \cite{Ooguri:2006in} and the Emergent String Conjecture~\cite{Lee:2019wij}. We will see that within the specific setup of vector multiplet moduli spaces of four-dimensional $\cN=2$ supergravities, the extended supersymmetry allows for a precise identification of those pieces of the conjectures which are a result of supersymmetry alone and those which contain genuine information on the UV completion of the theory. Recall that the Distance Conjecture proposes that at infinite distance in the moduli space of any effective description of quantum gravity an infinite tower of states becomes light exponentially fast in the geodesic distance,
\begin{equation*}
    \frac{m_{\rm tower}}{M_{\rm Pl}}\sim\exp(-\alpha\triangle)\,,
\end{equation*}
thereby triggering the breakdown of this effective description. This tower can always be interpreted (possibly in some dual description) as consisting of Kaluza--Klein modes or the excitations of a fundamental critical string, at least according to the Emergent String Conjecture. In other words, every infinite distance limit in a quantum gravity moduli space should either be a decompactification or an emergent string limit. Evidence for these two conjectures has been found both in top-down approaches (including \cite{Grimm:2018ohb,Blumenhagen:2018nts,Lee:2018urn,Grimm:2018cpv,Corvilain:2018lgw,Lee:2019xtm,Lee:2019wij,Grimm:2019bey,Baume:2019sry,Xu:2020nlh,Heidenreich:2020ptx,Lanza:2021udy,Klaewer:2020lfg,Bastian:2020egp,Grimm:2021ikg,Bastian:2021eom,Alvarez-Garcia:2021pxo,Etheredge:2022opl,Rudelius:2023odg,Etheredge:2023odp,Alvarez-Garcia:2023gdd,Alvarez-Garcia:2023qqj,Aoufia:2024awo,Hassfeld:2025uoy,Monnee:2025ynn,Monnee:2025msf,Aoufia:2026bau}) and from the bottom-up perspective (such as \cite{Hamada:2021yxy,Calderon-Infante:2023ler,Cribiori:2023ffn,Basile:2023blg,Basile:2024dqq,Herraez:2024kux,Bedroya:2024ubj,Kaufmann:2024gqo,david-vicente,Baines:2025upi,Baines:2026aug}). To analyse the role played by supersymmetry, it is important to carefully distinguish \cite{Kaufmann:2024gqo,Monnee:2025ynn} the following two aspects of the Emergent String Conjecture:
\begin{enumerate}
    \item The dichotomy of infinite distance limits as being either of decompactification or of emergent string type \cite{Lee:2019wij} poses  strong constraints on the exponential rates $\alpha$ appearing in the tower mass scale~\cite{Agmon:2022thq,Etheredge:2022opl,Castellano:2023stg,Castellano:2023jjt,Etheredge:2024tok,Etheredge:2025ahf}. In limits where the scaling parameters go to infinity  at the same rate, $\alpha$ is constrained as 
    \begin{equation} \label{scaling-gen}
        \alpha = \alpha_{\rm KK} = \sqrt{\frac{1}{d-2} + \frac{1}{n}} \,  \qquad {\rm or} \qquad \alpha =\alpha_{\rm string} = \sqrt{\frac{1}{d-2}} \,.
    \end{equation}
    Here the first possibility corresponds to a decompactification limit from $d$ to $d+n$ spacetime dimensions, while the second case refers to an asymptotic weak coupling limit of an emergent string. In more general multi-parameter limits 
    \begin{equation}  \label{gen-bound}
        \alpha \geq \sqrt{\frac{1}{d-2}}\,.
    \end{equation}
 
    \item Both the tower spacing and the degeneracy of the leading asymptotically light tower is constrained according to its interpretation as a KK/winding or a string excitation tower:
    \begin{eqnarray*}
        & {\rm KK:} \quad m_k = k \, m_0    \,, \qquad &d_k \sim k^{n-1}  \label{KK-tower}\,, \\
        & {\rm string:} \quad m_k = \sqrt{k} \, m_0    \,, \qquad &d_k \sim e^{c \sqrt{k}} \,.\label{string-tower}
    \end{eqnarray*}
\end{enumerate}
Note that in many contexts, the scalings \eqref{scaling-gen} can equivalently be interpreted as the vanishing rates of suitable gauge couplings. As demonstrated in \cite{Kaufmann:2024gqo} in a 5d context, this makes it possible to analyse the scalings \eqref{scaling-gen} already at the level of supergravity without having to refer to a particle tower, which is ultra-violet information beyond the original infrared theory.

The key observation of the present note is that for vector multiplet moduli spaces of 4d $\cN=2$ supergravities, the scaling behaviour \eqref{scaling-gen} follows as a direct consequence of the supersymmetry algebra. More precisely, in any single-parameter infinite distance limit in the vector multiplet moduli space of a 4d $\cN=2$ supergravity theory, the gauge couplings vanishing at the fastest rate exhibit the characteristic scaling \eqref{scaling-gen} with $n=1$, $n=2$ or $n=\infty$ (reproducing the scaling of $\alpha_{\rm string}$). Furthermore, in all multi-parameter limits of simple normal crossing type, the asymptotically vanishing gauge couplings can be shown to respect the scalings \eqref{scaling-gen} or \eqref{gen-bound}, as appropriate. Intriguingly, this behaviour follows without further assumptions on the UV completion of the low-energy supergravity theory. Indeed, revisiting early results on 4d $\cN=2$ supergravity~\cite{Ferrara:1991np,Ceresole:1992su} (see also \cite{Cecotti:1989kn,Cecotti:2020rjq}), we stress the existence of a {\it real} variation of (weight 3) polarised Hodge structure (VPHS) underlying all 4d $\cN=2$ supergravity vector multiplet moduli spaces. Thanks to very recent advances in the mathematical literature~\cite{Deng2022OnTN,complexVHS,complexVHS-multi}, this real VPHS guarantees the above scaling behaviour of gauge couplings for limits of simple normal crossing type.

Examples of real VPHS are those governing the vector multiplet moduli space of compactifications of Type II string theory on Calabi--Yau threefolds. All such VPHS have the important additional property that they descend from an {\it integral} VPHS, which is non-trivial extra structure not present, a priori, in an arbitrary 4d $\cN=2$ supergravity theory. This integrality of the subset of geometrically realised VPHS underlies the classic mathematical results of \cite{schmid, CKS}. In the quantum gravity literature, the technology of integral VPHS was introduced in the context of Type IIB Calabi--Yau compactifications in~\cite{Grimm:2018ohb} and has played an important role ever since~\cite{Grimm:2018cpv,Corvilain:2018lgw,Grimm:2019wtx,Grimm:2019bey,Grimm:2020cda,Bastian:2020egp,Grimm:2021ikg,Bastian:2021eom,Bastian:2021hpc,Grimm:2021ckh,Hassfeld:2025uoy,Monnee:2025ynn,Monnee:2025msf}. In particular, the scaling behaviour \eqref{scaling-gen} can be deduced from such integral VPHS in string compactifications, while the existence and properties of the relevant particle towers only follow from more detailed properties of the degenerating geometry \cite{Lee:2019wij,Hassfeld:2025uoy,Monnee:2025ynn,Monnee:2025msf}. However, the characteristic exponential scaling rates \eqref{scaling-gen} are not tied to a geometric realisation of the VPHS, and in fact not to any UV completion of the 4d $\cN=2$ supergravity. This perhaps surprising result is due to the fact that already for the real VPHS observed in \cite{Ferrara:1991np,Ceresole:1992su} in general 4d ${\cal N}=2$ supergravities, the characteristic asymptotic scaling of the Hodge norm  \cite{schmid,CKS} can be established \cite{Deng2022OnTN,complexVHS,complexVHS-multi}. This is true for single-parameter limits without any specification \cite{complexVHS} and for multi-parameter limits under the assumption of simple normal crossing \cite{complexVHS-multi} (see also \cite{Deng2022OnTN}). Combined with the previous works~\cite{Grimm:2018ohb,Grimm:2018cpv,Corvilain:2018lgw,Grimm:2019wtx,Grimm:2019bey,Grimm:2020cda,Bastian:2020egp,Grimm:2021ikg,Bastian:2021eom,Bastian:2021hpc,Grimm:2021ckh,Hassfeld:2025uoy,Monnee:2025ynn,Monnee:2025msf}, this leads to the scaling laws \eqref{scaling-gen} in the vector multiplet sector of general 4d $\cN=2$ supergravity theories.
 
This does by no means imply that every such supergravity theory can arise from a consistent quantum gravity in the UV. In fact, we argue that two of the important extra properties which are automatic for VPHS of Calabi--Yau (and more generally geometric) origin, integrality and simple normal crossing, can be related to universal principles of quantum gravity and are hence necessary for consistency of the theory: First, the existence of an integral structure underlying the real VPHS provided by the supergravity (Quantum Gravity Principle \ref{Swamp1}) is directly linked to the completeness hypothesis. Second, the need for well-defined duality frames at infinite distance necessitates the existence of a (simple) 
normal crossing compactification of moduli space (Quantum Gravity Principle \ref{Swamp2}). 

We stress that we are not the first to discuss the VPHS underlying 4d $\cN=2$ vector multiplet moduli spaces \cite{Cecotti:1989kn,Ferrara:1991np,Ceresole:1992su, Cecotti:2020rjq} in the context of quantum gravity. It was already anticipated in~\cite{Cecotti:1989kn} that the lack of an integral structure underlying this VPHS should be related to properties of the quantum theory. We identify this quantum datum as the completeness hypothesis put forward in~\cite{Polchinski:2003bq,Banks:2010zn} and emphasise that this additional property of the VPHS is also an obvious necessary condition for the existence of the Distance Conjecture tower of states.  More subtle is the importance of a simple normal crossing compactification of moduli space for compatibility with the Emergent String Conjecture. Note that the normal crossing condition arises also in the mathematical work \cite{complexVHS-multi} as a sufficient condition to establish the vanishing theorems of the Hodge norm underlying the scaling behaviour \eqref{scaling-gen}. Our main point is that this condition is also necessary for physical reasons: Specifically, it is key to guarantee uniqueness \cite{Lee:2019wij} of the fundamental string in emergent string limits (see \cite{Hassfeld:2025uoy} for a discussion in the complex structure moduli space of Type IIB compactifications) and more generally to allow for peaceful coexistence of the asymptotic gravitational duality frames associated with the infinite distance strata coming together at higher codimension in moduli space. In this sense, integrality and (simple) normal crossing are necessary conditions, at the supergravity level, for the appearance of light gravitational towers as predicted by the Distance and Emergent String Conjecture, which represents the genuine quantum gravity constraint on the 4d ${\cal N}=2$ supergravity.
 
This article is structured as follows. In Section~\ref{sec:special-geometry} we review the projective special K\"ahler geometry of 4d $\cN=2$ vector multiplet moduli spaces and give an explicit construction of the real variation of polarised Hodge structure which is intrinsic to this supergravity data. In Section~\ref{sec:Asymptotics} we redo (parts of) the analysis of~\cite{Grimm:2018cpv,Grimm:2018ohb,Hassfeld:2025uoy,Monnee:2025ynn} in the more general supergravity setting and interpret the resulting classification of infinite distance limits in the light of the Emergent String Conjecture. In Section~\ref{sec:swamp} we propose integrality of the VPHS and (simple) normal crossing boundaries of moduli space as necessary quantum gravity principles by linking them to the completeness hypothesis and the concept of dualities, respectively. We also compare our findings to the class of effective theories whose moduli spaces are symmetric spaces, as analysed in \cite{Baines:2025upi,Baines:2026aug}. In Section~\ref{sec:positive} we analyse in detail supergravities with a cubic prepotential which arises near points of maximal unipotent monodromy and establish a basis in which all cubic coefficients in the prepotential are non-negative. Finally, we conclude in Section~\ref{sec:discussion}.

\section{Projective special K\"ahler geometry}\label{sec:special-geometry}
A four-dimensional supergravity theory preserving eight real supercharges coupled to $n$ vector multiplets is described by the action~\cite{Andrianopoli_1997}
\begin{equation}\label{eq:4d-action}
    S_{\rm 4d}=\int_{\mathbb{R}^{1,3}}\left(\frac{M_\mathrm{Pl}^2}{2}\star R
    -2g_{i\bar\jmath}\mathrm{d}t^i\wedge\star\mathrm{d}\overline{t}^{\bar\jmath}
    +\frac{1}{2}\mathrm{Im}\mathcal{N}_{IJ}(t,\bar{t})F^I\wedge\star F^J 
    +\frac{1}{2}\mathrm{Re}\mathcal{N}_{IJ}(t,\bar{t})F^I\wedge F^J
     \right)\,,
\end{equation}
where the $t^i$, $i=1,\dots,n$, denote the complex scalars in the vector multiplets and $F^I={\rm d}A^I$, $I=0,\dots,n$, label the field strengths of the ${\rm U}(1)$ gauge fields in the vector multiplets as well as the graviphoton. In the above we ignore the hypermultiplet sector and restrict ourselves to the two-derivative level of the bosonic part of the action. 
 
The vector multiplet moduli space $M_n$ associated to the action~\eqref{eq:4d-action} is a projective special K\"ahler manifold of complex dimension $n$~\cite{deWit:1984wbb,Castellani:1990tp,Castellani:1990zd,Strominger:1990pd,Craps:1997gp} with local coordinates $t^i$. In the following, we will review two standard formulations of this geometry and refer to~\cite{Craps:1997gp} for a proof of their equivalence. While the local definition we state first is more commonly used in practice, the second more global definition turns out to be very useful to describe the structure underlying the asymptotic regimes of $M_n$.

\subsection{Local formulation}\label{ssec:local}
A projective special K\"ahler manifold $M_n$ of dimension $n$ is a K\"ahler manifold (with K\"ahler metric $g_{i\bar \jmath}$ as appearing in~\eqref{eq:4d-action}) such that around each point of $M_n$ there is a set of local homogeneous coordinates $Z^I$, $I=0,\dots,n$ with the following properties: There exists a holomorphic prepotential function $\mathbb{z}(Z^I)$, homogeneous of degree two in the $Z^I$, such that the K\"ahler potential $K$ on $M_n$ can be written as
\begin{equation}\label{eq:sg-Kpot}
    e^{-K}=2{\rm Im}\left(Z^I\overline{\mathbb{z}}_I\right)\,.
\end{equation}
Here $\mathbb{z}_I=\partial_I\mathbb{z}$ are the derivatives of $\mathbb{z}$ with respect to the homogeneous coordinates $Z^I$.\footnote{\label{fnote-homo}For future purposes we recall that homogeneity of $\mathbb{z}$, i.e. $Z^J \mathbb{z}_J = 2 \mathbb{z}$, implies the useful identities $\mathbb{z}_I = \mathbb{z}_{IJ} Z^J$ and $\mathbb{z}_{IJK}Z^K =0$.} The latter are related to the (inhomogeneous) local coordinates $t^i$ of $M_n$ as  $t^i = {Z^i/Z^0}$. In a physics context, the prepotential also determines the couplings appearing in the action~\eqref{eq:4d-action}  via 
\begin{equation}\label{eq:4d-gaugekin}
    \cN_{IJ}=\overline{\mathbb{z}}_{IJ}+2\ii\frac{\Im\mathbb{z}_{IK}Z^L\,\Im\mathbb{z}_{JL}Z^L}{Z^M\Im\mathbb{z}_{MN}Z^N}\,.
\end{equation}
The quantities $Z^I$ and $\mathbb{z}_I$ are combined in the so-called period vector
\begin{equation}\label{eq:period-vector}
    \Pi=\left(\begin{array}{c}
        Z^I\\ \mathbb{z}_I
    \end{array}\right)\,.
\end{equation}
Moreover, between each two overlapping charts $U_\alpha, U_\beta\subset M_n$, the period vector transforms as
\begin{equation}\label{eq:sg-local-trafo}
    \Pi_{(\alpha)}=e^{-f_{\alpha\beta}(Z)}M_{\alpha\beta}\Pi_{(\beta)}\,,
\end{equation}
where $f_{\alpha\beta}:U_\alpha\cap U_\beta\to\mathbb{C}$ is a holomorphic function and $M_{\alpha\beta}\in{\rm Sp}(2n+2,\mathbb{R})$. The transition functions $e^{-f_{\alpha\beta}}M_{\alpha\beta}$ have to satisfy the usual cocycle conditions. We refer to the review~\cite{Craps:1997gp} for additional details on the derivation of this structure from the closure of the 4d $\cN=2$  supersymmetry algebra. The latter may be expressed as a set of four coupled non-holomorphic first order differential equations for the period vector $\Pi$~\cite{Ferrara:1991np,Ceresole:1992su},
\begin{align}\label{eq:non-holo}
    \begin{split}
        D_i \Pi&=U_i\,,\\
        D_i U_j&=-\ii W_{ijk}g^{k\bar k}\bar{U}_{\bar{k}}\,,\\
        D_i\bar{U}_{\bar\jmath}&=g_{i\bar\jmath}\bar{\Pi}\,,\\
        D_i\bar\Pi&=0\,.
    \end{split}
\end{align}
Here, as before $i=1,\dots,n$, 
\begin{equation}\label{eq:K-deriv}
    D_i=\partial_i+p\partial_iK
\end{equation}
is the K\"ahler covariant derivative acting locally on a function of K\"ahler weight $p$\footnote{More specifically, a function $g:M_n\to\mathbb{C}$ has (chiral) K\"ahler weights $(p,\bar{p})$ if under a K\"ahler transformation $K(z,\bar{z})\to K(z,\bar{z})+f(z)+\bar{f}(\bar{z})$ with $f$ holomorphic, $g$ transforms as $g\to e^{-pf-\bar{p}\bar{f}}g$. For example, the holomorphic period vector $\Pi$ has K\"ahler weights $(1,0)$, see~\eqref{eq:sg-local-trafo}, while $e^{K/2}\Pi$ has weights $(\frac{1}{2},-\frac{1}{2})$.\label{fn:Kahler-weight}} and the fully symmetric tensor $W_{ijk}$ is related to the third derivatives of the prepotential via
\begin{equation}\label{eq:sg-yukawa}
    W_{ijk}=e^K\partial_iZ^I\partial_jZ^J\partial_kZ^K\mathbb{z}_{IJK}\,.
\end{equation}
In passing, we note that it was shown in~\cite{Ferrara:1991np,Ceresole:1992su} that the system~\eqref{eq:non-holo} may be rewritten as a single fourth order holomorphic differential equation of the schematic form
\begin{equation}\label{eq:holo}
    \left(\partial_i-\mathbf{A}_i\right)\mathbf{V}=0\,,
\end{equation}
where we refer to the cited references for details on the notation and the necessary algebra.

Notice that the transition functions of the period vector as written in~\eqref{eq:sg-local-trafo} contain an element of the group ${\rm Sp}(2n+2,\mathbb{R})$. This group is also referred to as the electric-magnetic duality group of the supergravity. Indeed, recall that the magnetic dual ${\rm U}(1)$ field strengths $G_I$ are given by
\begin{equation}\label{eq:magnetic-G}
    G_I=\frac{\delta S_{\rm 4d}}{\delta F^I}={\rm Im}(\cN_{IJ})\star F^J+{\rm Re}(\cN_{IJ})F^J\,.
\end{equation}
Assuming the existence of a (BPS) particle, its charge vector takes the form $q=(p^I,q_I)\in\cC\subseteq\mathbb{Z}^{2n+2}$, with the $p^I$ referring to its magnetic charges under $G_I$ and the $q_I$ are its electric charges with respect to the $F^I$. Such charge vectors are acted upon by an ${\rm Sp}(2n+2,\mathbb{Z})$ subgroup of the full electric-magnetic duality group ${\rm Sp}(2n+2,\mathbb{R})$, as follows from Dirac quantisation of charges. We will come back to this important point in Section~\ref{ssec:completeness}. For now, we work in the purely electric frame set by the action~\eqref{eq:4d-action} and do \emph{not} assume the existence of any kind of particle (or other object) coupled to the supergravity.

\subsection{Global formulation}\label{ssec:global}
Phrased more globally, the transformation properties~\eqref{eq:sg-local-trafo} show that the period vector $\Pi$ of the vector multiplet moduli space $M_n$ is a section of a holomorphic vector bundle of the form
\begin{equation}
    L\otimes H\to M_n\,,\quad \Pi\in H^0(M_n,L\otimes H)\,,
\end{equation}
over moduli space. The line bundle $L$ is a ${\rm U}(1)$-principal bundle and the vector bundle $H$ of rank ${\rm rk}_\mathbb{R}H=2(n+1)$ has structure group ${\rm Sp}(2n+2,\mathbb{R})$. The local relation~\eqref{eq:sg-Kpot} between the K\"ahler potential of $M_n$ and its prepotential can be stated in terms of the section $\Pi$ as
\begin{equation}\label{eq:Hodge-Kpot}
    e^{-K}=\ii\langle\bar\Pi,\Pi\rangle\,.
\end{equation}
Here the sesqui-linear form $\langle.,.\rangle$ on $L\otimes H$ is given in terms of the symplectic pairing $(\cdot,\cdot)$ on the fibers of $H$ and the Hermitian metric $h$ on $L$ as
\begin{equation}
    \langle\Pi,\Pi'\rangle_t=h(s,s')_t(v,v')_t\,,
\end{equation}
for $\Pi(t)=s\otimes v\in (L\otimes H)_t$ and likewise for $\Pi'$. Furthermore, from the form of the K\"ahler covariant derivative~\eqref{eq:K-deriv} entering the local defining equations~\eqref{eq:non-holo} it follows that the connection 1-form on $L$ is the derivative of the K\"ahler potential $K$ so that the first Chern class of $L$ coincides with the class of the K\"ahler form $J=\frac{\ii}{2\pi}\partial\bar\partial K$ on $M_n$. In other words, $[J]\in H^2(M_n,\mathbb{Z})$.\footnote{A K\"ahler manifold with integral K\"ahler class is called Hodge--K\"ahler manifold.} The line bundle $L$ reflects the projectivity of the section $\Pi$ and is related to K\"ahler transformations in the sense that a rescaling $\Pi\to e^{-f(z)}\Pi$ induces the transformation
\begin{equation}\label{eq:K-trafo}
    K\to K+f+\bar f
\end{equation}
of the K\"ahler potential, see also Footnote~\ref{fn:Kahler-weight}. For details on the equivalence of the local and global formulations of projective special K\"ahler geometry we refer to~\cite{Craps:1997gp}.

We close this section by noting that there is yet another definition of local/ projective special K\"ahler manifolds, namely as quotients of so-called conic special K\"ahler manifolds. This notion is more commonly used in the math literature, see e.g.~\cite{Mantegazza:2021vyx,Freed:1997dp}. The construction presented in the next section of course works also from this perspective.

\subsection{An underlying variation of polarised Hodge structure}\label{ssec:VPHS}
This global viewpoint on special geometry makes the existence of an underlying Hodge structure apparent \cite{Ferrara:1991np,Ceresole:1992su} (see also \cite{Cecotti:1989kn,Cecotti:2020rjq}). To see this, we fix some point $t\in M_n$ and consider the fiber $H_{\mathbb{C},t}$ of the complexified bundle $H_\mathbb{C}=H\otimes\mathbb{C}$ over this point. Fixing a K\"ahler frame as in~\eqref{eq:K-trafo} (which in the local language of Section~\ref{ssec:local} means fixing $Z^0=1$), the section $\Pi\in H^0(M_n,L\otimes H)$ can be viewed as a section of $H$. By extending linearly to $\mathbb{C}$, we hence have a non-zero vector 
\begin{equation}
     V = e^{K(t,\bar{t})/2}\Pi(t)\in H_{\mathbb{C},t}
\end{equation}
in the complex fiber $H_{\mathbb{C},t}$ over $t\in M_n$.\footnote{To follow the conventions of~\cite{Craps:1997gp}, in the following we use the covariantly holomorphic section $e^{K/2}\Pi$ instead of the holomorphic $\Pi$.} In a local trivialisation of $H$ near $t\in M_n$, the K\"ahler covariant derivative of $V$ then takes the form $D_iV=\partial_iV+\frac{1}{2}(\partial_iK)V$, where $\partial_i=\frac{\partial}{\partial t^i}$ for local coordinates $t^i$ near $t$ (see also Footnote~\ref{fn:Kahler-weight}). Using the standard form of the complexified symplectic pairing $(\cdot,\cdot)$ on the fiber $H_{\mathbb{C},t}$ as well as the homogeneity of the prepotential (see in particular Footnote \ref{fnote-homo}) it follows
 that
\begin{equation}\label{eq:SG-rels}
    (V,D_iV)= (V,\partial_iV)+\frac{1}{2}\partial_iK( V,V)=0\,,\,\,( D_iV,D_jV)=0\,.
\end{equation}

\paragraph{Hodge structure.} With these identities, we can now prove that there is a Hodge decomposition $H^{p,q}_t$ of the complexified fiber $H_{\mathbb{C},t}$. Recall that a Hodge structure of weight $w$ on a real vector space $H_t$ is a decomposition of its complexification $H_{\mathbb{C},t}$ into $w+1$ subspaces $H^{p,q}_t$ satisfying
\begin{equation}\label{eq:Hodge-decomp-def}
    H_{\mathbb{C},t}=\bigoplus_{p+q=w}H^{p,q}_t\,,\quad \overline{H^{p,q}_t}=H^{q,p}_t\,.
\end{equation}
As a candidate weight $w=3$ decomposition, we define
\begin{align}
    \begin{split}\label{eq:Hodge-decomp}
        H^{3,0}_t&={\rm span}_\mathbb{C}\{V\}\,,\\
        H^{2,1}_t&={\rm span}_\mathbb{C}\{D_iV\}\,,\\
        H^{1,2}_t&=\overline{H^{2,1}_t}={\rm span}_\mathbb{C}\{\overline{D_iV}\}\,,\\
        H^{0,3}_t&=\overline{H^{3,0}_t}\,.
    \end{split}
\end{align}
The only non-trivial statement to prove is that $H^{3,0}_t\oplus H^{2,1}_t$ has complex dimension $n+1$. As we will now show, this is a simple consequence of the identities~\eqref{eq:SG-rels}. Choose $\alpha,\beta^i\in\mathbb{C}$ such that $\alpha V+\beta^i D_iV=0$. Then, using~\eqref{eq:Hodge-Kpot},
\begin{equation}
    0=(\bar{V},\alpha V+\beta^i D_iV)=-\ii \alpha +(\bar{V},\beta^i D_iV)\,,
\end{equation}
where the second term vanishes due to $D_i(\bar{V},V)=0$.\footnote{Indeed, since $(\bar V, V) = i$ because of \eqref{eq:Hodge-Kpot}, we have  $0=D_i(\bar{V},V)=(D_i\bar{V},V)+(\bar{V},D_iV)$ and the first term vanishes as $\bar{V}$ is covariantly anti-holomorphic.} Therefore, $\alpha=0$ and $\beta^iD_iV=0$. Furthermore,
\begin{equation}
    0=(\beta^iD_iV,\overline{D_jV})=-\ii\beta^i g_{i\bar{\jmath}}\,,
\end{equation}
from which we conclude that also $\beta^i=0$.\footnote{A direct computation gives $[D_i,D_{\bar{\jmath}}]V=-g_{i\bar{\jmath}}V$. The desired equality then follows by taking a covariant derivative of $(\bar{V},D_iV)=0$.} It follows that~\eqref{eq:Hodge-decomp} is indeed a Hodge decomposition of the complexified fiber $H_{\mathbb{C},t}$. From this decomposition we may define the associated Hodge filtration $F^\bullet_t$ via
\begin{equation}
    F^p_t=\bigoplus_{q\geq p}H^{q,3-q}_t\,.
\end{equation}

\paragraph{Weil operator.} Associated to the Hodge structure on $H_t$ is a linear operator $C:H^{p,q}_t\to H^{p,q}_t$, called Weil operator, with the property \begin{equation}
Cv=\ii^{p-q}v
\end{equation}
for $v\in H^{p,q}_t$. We refer to the reviews~\cite{vandeHeisteeg:2022gsp,Monnee:2024gsq} for additional information. The important point for us is that $C$ can be written in terms of the gauge kinetic matrix $\cN$ given in~\eqref{eq:4d-gaugekin}. Indeed, defining~\cite{Ceresole:1995ca}
\begin{equation}
    \cM=\left(\begin{array}{cc}
        {\rm Im}(\cN)+{\rm Re}(\cN)({\rm Im}(\cN))^{-1}{\rm Re}(\cN) & -{\rm Re}(\cN)({\rm Im}(\cN))^{-1}  \\
        -({\rm Im}(\cN))^{-1}{\rm Re}(\cN) & ({\rm Im}(\cN))^{-1}
    \end{array}\right)\,,
\end{equation}
one checks that\footnote{To do so, one needs to use $\mathbb{z}_I=\cN_{IJ}Z^J$ as well as $D_i\mathbb{z}_I=\overline{\cN}_{IJ}D_iZ^J$. Both relations are proven in Section~2 of~\cite{Ceresole:1995ca} and we refer to this reference for more details.}
\begin{equation}\label{eq:Weil}
    C=\left(\begin{array}{cc}
         & -\mathbb{I}_{n+1} \\
        \mathbb{I}_{n+1} & 
    \end{array}\right)\cM
\end{equation}
satisfies the properties of the Weil operator associated with the given Hodge structure on $H_t$. With the help of $C$, one defines the Hodge norm
 \begin{equation}\label{eq:Hodge-norm}
    \|v\|^2=(v,C\bar{v})\,,\quad v\in H_{\mathbb{C},t}\,,
\end{equation}
which will play an important role in the sequel. For example, the relation \eqref{eq:Hodge-Kpot} can be expressed in terms of the Hodge norm as 
\begin{equation}
    e^{-K}=\|\Pi\|^2 \,.
\end{equation}

\paragraph{Variation of polarised Hodge structure.} If we drop the basepoint $t\in M_n$, we may consider how the Hodge decomposition of the fibers of $H_\mathbb{C}$ varies as we move in moduli space. The system~\eqref{eq:non-holo} of coupled differential equations can be rewritten in terms of the filtration $F^\bullet_t$ as
\begin{equation}
    \frac{\partial F^p}{\partial t^i}\subseteq F^{p-1}\,,\quad\frac{\partial F^p}{\partial\bar t^i}\subseteq F^p\,,
\end{equation}
which are exactly the Griffith's transversality and holomorphicity conditions underlying a variation of Hodge structure (VHS)~\cite{complexVHS,schmid}.\footnote{The original work~\cite{schmid} assumes the existence of an underlying integral structure, i.e. $H_t=H_{\mathbb{Z},t}\otimes_\mathbb{Z}\mathbb{R}$ for some $\mathbb{Z}$-module $H_{\mathbb{Z},t}$, whereas~\cite{CKS} generalises to real vector spaces but keeps an additional assumption on the monodromy. In the general supergravity setting this monodromy assumption does not hold which is why we rely on the more recent results of~\cite{complexVHS,complexVHS-multi}. In Section~\ref{ssec:completeness} we link the existence of an integral structure to the completeness hypothesis in quantum gravity. Such a connection was anticipated already in~\cite{Cecotti:1989kn}.} In this language, \eqref{eq:holo} is the associated Picard--Fuchs equation. In fact, the (anti-symmetric) symplectic pairing $(\cdot,\cdot)$ on the fibers of $H$ (and its complexification) is a polarisation for this VHS, meaning that
\begin{align}
    \begin{split}
        (H^{p,q}_t,H^{r,s}_t)=0\qquad & {\rm unless\,\,}(p,q)=(s,r)\,,\\
        \ii^{p-q}(v,\bar{v})>0\qquad & {\rm for\,\,}v\in H^{p,q}_t\setminus\{0\}\,.
    \end{split}
\end{align}
For example, for $(p,q)=(3,0)$ we have seen in previous paragraphs that $(V,D_iV)=0$, $(V,\overline{D_iV})=0$ and $\ii^3(V, \bar{V})=\ii^{-3}(\bar{V},V)=1$. The other cases are checked similarly. Thus, the symplectic pairing on $H$ turns the previously found VHS into a variation of polarised Hodge structure (VPHS).

To conclude, we have reviewed in this section the construction of a real weight 3 variation of polarised Hodge structure purely from the data given by a 4d $\cN=2$ supergravity theory coupled to a number of vector multiplets. We stress that this is a pure bottom-up argument and no further UV input---in particular no underlying geometric origin---has been assumed as everything is based solely on the closure of the supersymmetry algebra.

\section{Asymptotic couplings of 4d \texorpdfstring{$\cN=2$}{N=2} supergravities} \label{sec:Asymptotics}
We can now apply the machinery of asymptotic Hodge theory to study the behaviour of physical couplings at the asymptotic boundaries of the vector multiplet moduli space $M_n$ of a general 4d ${\cal N}=2$ supergravity theory. This analysis was carried out in great detail for moduli spaces arising from Calabi--Yau threefold compactifications of Type II string theory in~\cite{Grimm:2018ohb,Grimm:2018cpv,Corvilain:2018lgw,Grimm:2019wtx,Grimm:2019bey,Grimm:2020cda,Bastian:2020egp,Grimm:2021ikg,Bastian:2021eom,Bastian:2021hpc,Grimm:2021ckh,Hassfeld:2025uoy,Monnee:2025ynn,Monnee:2025msf}. The purpose of this section is to point out that at the level of couplings, this analysis carries over to the much more general 4d $\cN=2$ supergravity setting --- up to the important extra restriction to simple normal crossing in multi-parameter limits, which is automatically satisfied in limits arising from a Calabi-Yau compactification~\cite{viehweg}. 

The result of this discussion is
\begin{claim} \label{Claim1}
Consider a 4d ${\cal N}=2$ supergravity theory with vector multiplet moduli space $M_n$ and take an infinite distance limit on $M_n$ of simple normal crossing type. Then in any such limit, the gauge couplings behave like in a decompactification limit to 5d (type IV), to 6d (type III) or as in a weak coupling limit of fundamental string theory (type II). Specifically, the following properties hold:
\begin{enumerate}
    \item In every single-parameter limit or simple normal crossing multi-parameter limit in which all parameters scale at the same parametric rate (see \eqref{eq:homo-scaling}), there is a minimal physical 1-form coupling $\mathfrak{q}_{\rm min}$ (defined in \eqref{eq:1-form}) with the property that 
    \begin{equation}\label{eq:distance-scaling}
        \mathfrak{q}_{\rm min}^2\sim \exp\left(-\alpha_\mathtt{A}\triangle\right)\,
    \end{equation}
    in terms of the geodesic distance $\triangle$ on moduli space and with $\alpha_{\rm II}=1/\sqrt{2}$, $\alpha_{\rm III}=1$ and $\alpha_{\rm IV}=\sqrt{3/2}$.  In type III limits, there are two independent 1-form couplings leading to the scaling~\eqref{eq:distance-scaling} while in type IV limits this scaling is uniquely realised.
    \item In simple normal crossing multi-parameter limits in which some or all parameters scale at different rates, the exponential rate $\alpha_\mathtt{A}$ in \eqref{eq:distance-scaling} satisfies the bound $\alpha \geq \sqrt{\frac{1}{d-2}}$ for $d=4$ of \cite{Agmon:2022thq,Etheredge:2022opl}.
    \item In type III and IV limits, all physical 2-form couplings $\mathcal{Q}$ (defined in \eqref{eq:2-form}) satisfy the hierarchy $\mathcal{Q}\succ\mathfrak{q}^2_{\rm min}$, while in type II limits, there exists a minimal physical 2-form coupling $\mathcal{Q}_{\rm min}$ with the property $\mathfrak{q}^2_{\rm min}\sim\mathcal{Q}_{\rm min}$. Here $\mathcal{Q}_{\rm min}$ is uniquely determined in all 1-parameter limits, while in simple normal crossing multi-parameter limits several such $\mathcal{Q}_{\rm min}$ exist, which are mutually local with respect to each other by duality. 
 \end{enumerate}
\end{claim}

We will now justify this claim by extending the reasoning of~\cite{Grimm:2018ohb,Grimm:2018cpv,Corvilain:2018lgw,Grimm:2019wtx,Grimm:2019bey,Grimm:2020cda,Bastian:2020egp,Grimm:2021ikg,Bastian:2021eom,Bastian:2021hpc,Grimm:2021ckh,Hassfeld:2025uoy,Monnee:2025ynn,Monnee:2025msf} to general 4d ${\cal N}=2$ supergravity theories. The basis for this extension are the recent mathematical works \cite{complexVHS,Deng2022OnTN,complexVHS-multi}. These prove the classical results of \cite{schmid,CKS} underlying the analysis of limits in Calabi-Yau compactifications for more general complex variations of Hodge structure without quasi-unipotent monodromy as realised in general 4d ${\cal N}=2$ supergravity theories. We analyse one-parameter limits and simple normal crossing multi-parameter limits in turn.

\subsection{Single-parameter limits}\label{ssec:single}
Infinite distance limits in the vector multiplet moduli space $M_n$ correspond to boundaries of $M_n$. In the language of asymptotic Hodge theory, these boundaries turn out to be singularities at which the Hodge decomposition~\eqref{eq:Hodge-decomp} of the fibers of $H_\mathbb{C}\to M_n$ breaks down~\cite{schmid,CKS,complexVHS,complexVHS-multi,Grimm:2018ohb,vandeHeisteeg:2022gsp,Monnee:2024gsq}. Each one-parameter limit leads to a single component of the boundary divisor of $M_n$ and is therefore referred to as a codimension-one singularity. Multi-parameter limits correspond to the intersection of several (one-parameter) boundary components and hence to a singularity of higher codimension. They will be studied in the next section.

Around a one-parameter singularity, the relevant local model is the punctured disc $\Delta^\ast\subset\mathbb{C}$, with $0\in\Delta$ corresponding to the singularity. For $z$ the local coordinate on $\Delta^\ast$, we can consider the monodromy action $z\to e^{2\pi\ii}z$, which defines a monodromy operator $T$ acting on the entries of the Hodge filtration $F^\bullet$. The Jordan decomposition of $T$ reads $T=T_sT_u$ with a semi-simple part $T_s$ and a unipotent part $T_u$. In absence of an underlying integral structure, the eigenvalues of $T$ are only known to lie on the unit circle but do not necessarily have to be roots of unity~\cite{complexVHS}. In other words, the semi-simple part $T_s$ can have infinite order and is therefore not removable by a finite base change as done in~\cite{schmid,CKS}. A priori, this introduces a dependence of the limiting mixed Hodge structure at $z=0$ not only on the unipotent part $T_u$, but also on the semi-simple part $T_s$ of the monodromy operator.\footnote{For details on the precise definition and construction of the limiting mixed Hodge structure we refer to the original works~\cite{schmid,CKS,complexVHS,complexVHS-multi} as well as the reviews~\cite{Monnee:2024gsq,vandeHeisteeg:2022gsp}.} However, as shown in~\cite{complexVHS}, 
the limiting mixed Hodge structure still turns out to be independent of the semi-simple part $T_s$ and is therefore fully characterised by the unipotent part $T_u=e^N$ with $N$ the nilpotent log-monodromy matrix. The important upshot of this technical point is that even without an underlying integral structure, there is a Deligne splitting of $H_\mathbb{C}$ associated with the limiting mixed Hodge structure,
\begin{equation}\label{eq:Deligne}
    H_\mathbb{C}=\bigoplus_{0\leq p,q\leq 3}I^{p,q}\,,
\end{equation}
where the dimensions $\dim_\mathbb{C}I^{p,q}=i^{p,q}$ satisfy $h^{p,3-p}=\sum_q i^{p,q}$ and further symmetry relations, which follow from the explicit construction~\cite{CKS,vandeHeisteeg:2022gsp,Monnee:2024gsq,complexVHS,complexVHS-multi}. Since $h^{3,0}=1$ (recalling from \eqref{eq:Hodge-decomp} that by construction, $H^{3,0}_t={\rm span}_\mathbb{C}\{V\}$), it follows that only a single $i^{3,q}$ can be non-vanishing. This gives four types of limits known as types I,\dots, IV for $q=0,\dots,3$, familiar in the string theory literature from the special case of Calabi--Yau compactifications \cite{Grimm:2018cpv,Grimm:2018ohb,vandeHeisteeg:2022gsp,Monnee:2024gsq}. In the following we focus on the infinite distance limits of type II, III and IV. We can immediately borrow the analogous results from the Calabi--Yau case, whose most essential points we briefly collect here for the reader's benefit. We refer to~\cite{vandeHeisteeg:2022gsp,Monnee:2024gsq} and references therein for details and derivations.

From the additional symmetries of the $i^{p,q}$ it follows that there is only a single further independent dimension, which by convention is chosen to be $i^{2,2}$. This denotes the secondary singularity type of the degeneration.\footnote{This finer distinction will not be of relevance to us. See~\cite{Hattab:2025aok,Monnee:2025msf,Hattab:2026lho,cgm} for recent discussions on subtleties related to type II$_0$ loci.} Furthermore, the log-monodromy matrix can be seen as a map~\cite{schmid}
\begin{equation}\label{eq:log-monodromy}
    N:I^{p,q}\to I^{p-1,q-1}\,,
\end{equation}
showing that $N^{d+1}=0$ for $d\leq 3$. More precisely, for limits of type II, III, IV, we have
\begin{equation} \label{eq:d-values}
d_{\rm II}=1 \,, \qquad d_{\rm III}=2 \, \qquad  d_{\rm IV}=3\,.
\end{equation}
Based on the Deligne splitting~\eqref{eq:Deligne} one defines the graded spaces ${\rm Gr}_\ell$ as 
\begin{equation}\label{eq:graded}
    {\rm Gr}_\ell=\bigoplus_{p+q=\ell}I^{p,q}\,,
\end{equation}
each of which carries a pure Hodge structure of weight $\ell$~\cite{CKS,complexVHS}. As we will see, of special importance for our analysis are the spaces ${\rm Gr}_2, {\rm Gr}_1, {\rm Gr}_0$, which for the different types of limits are of the following dimensions (with parameters $b$, $c$ and $d$ identified with $i^{2,2}$)~\cite{Grimm:2018cpv}:

\begin{eqnarray}\begin{split}
    {\rm type \,\,  II}_b: \quad   & \dim_\mathbb{C}({\rm Gr}_2) = 2 +b \,, & \quad  \dim_\mathbb{C}({\rm Gr}_1)=0 \,, \quad  \dim_\mathbb{C}({\rm Gr}_0) = 0   \,, \\
    {\rm type \,\,  III}_c: \quad   & \dim_\mathbb{C}({\rm Gr}_2) = c \,, & \quad  \dim_\mathbb{C}({\rm Gr}_1) = 2 \,, \quad  \dim_\mathbb{C}({\rm Gr}_0) = 0   \,, \\
    {\rm type \,\,  IV}_d: \quad   & \dim_\mathbb{C}({\rm Gr}_2) = d \,, &\quad  \dim_\mathbb{C}({\rm Gr}_1) = 0 \,, \quad  \dim_\mathbb{C}({\rm Gr}_0) = 1   \,.
\end{split} \label{eq:LimGr}
\end{eqnarray}
In particular, the minimal value of $\ell$ such that $\rm Gr_{\ell_{\rm min}} \neq 0$ is given by 
\begin{equation} \label{eq:lmin}
\ell_{\rm min} = 3 - d
\end{equation}
for $d$ as listed in \eqref{eq:d-values}.

\paragraph{Asymptotic K\"ahler potential.} The local model near a codimension one singularity in $M_n$ is the punctured disk $\Delta^\ast$ with local coordinate $z$. A single-parameter limit $z\to0$ therefore corresponds to a codimension one singularity in moduli space $M_n$ along the divisor $D=\{z=0\}$. The local coordinate $z$ is related to the physical scalars $t^i$ entering the action~\eqref{eq:4d-action} via
\begin{equation}\label{eq:covering-coords}
    t^i=\frac{e^i}{2\pi\ii}\log(z)\,,
\end{equation}
where the role of the vector $e=(e^i)\in\mathbb{R}^n$ is explained around~\eqref{eq:2-form}. For single-parameter infinite distance limits, only a single entry $e^1>0$ of the vector $e$ is non-vanishing and has furthermore to be positive. The corresponding vector multiplet scalar is $t^1$ and the divisor $D=\{z=0\}$ is reached as $t^1\to\ii\infty$. In this regime, the period vector $\Pi$ from~\eqref{eq:period-vector} can be approximated using the nilpotent orbit theorem~\cite{complexVHS,schmid} as
\begin{equation}\label{eq:nilpotent-single}
    \Pi=e^{Nt^1}\left(\mathbf{a}_0+\sum_{k}\mathbf{a}_k e^{2\pi\ii k t^1}\right)\,,
\end{equation}
where the coefficient vectors $\mathbf{a}_k$ depend on the coordinates $t^j$ with $j>1$ which stay finite along the limit. Since $N$ is nilpotent of order $d\leq4$, recall~\eqref{eq:log-monodromy}, the first exponential is in fact a polynomial in $t^1$ of degree at most three. Together with~\eqref{eq:sg-Kpot} and~\eqref{eq:Hodge-Kpot} this implies that any polynomial part in both the K\"ahler potential and prepotential can be at most cubic in $t^1$. We will discuss the cubic prepotential in more detail in Section~\ref{sec:positive}. Furthermore, via~\eqref{eq:Hodge-Kpot}, the K\"ahler potential $K$ is independent of the real scalar ${\rm Re}(t^1)$ to polynomial order in $t^1$, i.e. up to exponentially suppressed terms.\footnote{In fact, $K$ stays independent of ${\rm Re}(t^1)$ also when including the exponential terms in~\eqref{eq:nilpotent-single}. This follows because the inner product used in~\eqref{eq:Hodge-Kpot} is monodromy invariant and, as we discuss in more detail around~\eqref{eq:2-form}, the scalar ${\rm Re}(t^1)$ encodes this monodromy. The expansion~\eqref{eq:nilpotent-single} is of course valid only asymptotically.} Accordingly, this scalar enjoys a shift symmetry close to the asymptotic regime and is therefore referred to as an axion.

\paragraph{Scaling of physical couplings.} To study the physics associated to single-parameter infinite distance limits in $M_n$, the formulation of the boundary data in terms of the graded spaces \eqref{eq:graded} is particularly useful~\cite{Grimm:2018ohb,Grimm:2018cpv,Grimm:2019bey,Grimm:2019wtx,Hassfeld:2025uoy,Monnee:2025ynn}. This is because the growth theorem of~\cite{CKS,complexVHS} states that near $t^1\to\ii\infty$, 
\begin{equation}
\label{eq:growth-thm-single}
    \|q\|^2\sim{\rm Im}(t^1)^{\ell-3}\, \qquad {\rm for} \,  q\in{\rm Gr}_\ell \,,
\end{equation}
where $\|\cdot\|$ denotes the Hodge norm introduced in~\eqref{eq:Hodge-norm}. In particular, the Hodge norms vanishing at the fastest rate for ${\rm Im}(t^1) \to \infty$ are associated with $q \in {\rm Gr}_{\ell_{min}}$ with $\ell_{min}=2,1,0$ for type II, III, IV limits (see \eqref{eq:lmin}). 

To connect this to the couplings appearing in the action~\eqref{eq:4d-action}, recall from the discussion following~\eqref{eq:magnetic-G} that the symplectic bundle $H\to M_n$ encodes the splitting into electric and magnetic gauge sectors. In particular, the action~\eqref{eq:4d-action} only depends on the electric gauge fields $A^I$. An element $q\in H_{\mathbb{C},t}\cap\mathbb{R}^{2n+2}$ in the fiber of $H_\mathbb{C}\to M_n$ over $t\in M_n$ can naturally be written as $q=(p^I,q_I)$ with the first $n+1$ entries labeling the magnetic gauge fields $G_I$ and the latter $n+1$ entries the electric gauge fields $F^I$. Following~\cite{Ceresole:1995ca}, see also~\cite{Palti:2017elp,Bastian:2020egp}, we define the {\it physical 1-form coupling} $\mathfrak{q}_q$ associated to an element $q\in H_{\mathbb{C},t}\cap\mathbb{R}^{2n+2}$ as 
\begin{equation}\label{eq:1-form}
    \mathfrak{q}_q^2=-\frac{1}{2}q^T\cM q=\frac{1}{2}(q,Cq)=\frac{1}{2}\|q\|^2\,.
\end{equation}
For example, working with the action~\eqref{eq:4d-action}, for purely electric vectors $q=(0,q_I)$  the physical 1-form coupling takes the form
\begin{equation}
    \mathfrak{q}_q^2=-\frac{1}{2}q_I\left(\left({\rm Im}\mathcal{N}\right)^{-1}\right)^{IJ}q_J\,.
\end{equation}
As recalled around~\eqref{eq:nilpotent-single}, the prepotential $\mathbb{z}$ near the infinite distance singularity $t^1\to\ii\infty$ in $M_n$ is independent of the axion ${\rm Re}(t^1)$. In these regimes, the axion ${\rm Re}(t^1)$ may be dualised into an ``electric'' two-form potential $B_1$. Similar to~\eqref{eq:1-form}, we define the associated {\it physical 2-form coupling} as
\begin{equation}\label{eq:2-form}
    \mathcal{Q}_e^2=e^iG_{ij}e^j\,,
\end{equation}
where $G=(G_{ij})$ is the moduli space metric written in terms of the real coordinates ${\rm Re}(t^i),\,{\rm Im}(t^i)$. The vector $e=(e^i)\in\mathbb{R}^n$, introduced in~\eqref{eq:covering-coords}, collects the shift(s) of the axion ${\rm Re}(t^1)$ around the singularity $z=0$. In other words, under the monodromy $z\to e^{2\pi\ii}z$ the axion ${\rm Re}(t^1)$ shifts as $t^1\to t^1+e^1$.

In view of the growth theorem~\eqref{eq:growth-thm}, the splitting~\eqref{eq:Deligne} of $H_\mathbb{C}$ into graded spaces in the infinite distance limit reflects the possible scaling behaviour of the associated physical 1-form couplings. The scaling of physical 2-form couplings can likewise be determined from the growth theorem via
\begin{equation}
    e^{-K}=\|\Pi\|^2\sim {\rm Im}(t^1)^{d}\,,
\end{equation}
for the values $d=1,2,3$ (see \eqref{eq:d-values}) for a degeneration of type II, III, IV. It follows that
\begin{equation}
    \cQ_{\rm min}^2\sim G_{11}=\frac{1}{4}\frac{\partial^2 K}{\partial\left({\rm Im}(t^1)\right)^2}
    =\frac{1}{4}\frac{d}{\left({\rm Im}(t^1)\right)^2}\,.
\end{equation}
Compared with the scaling of $\mathfrak{q}_{\rm min}$ given in~\eqref{eq:1-form} and~\eqref{eq:growth-thm-single} one finds that all single-parameter limits realise the parametric hierarchy
\begin{equation}\label{eq:hierarchy-couplings}
    \mathcal{Q}_{\rm min}\succsim\mathfrak{q}^2_{\rm min}\,.
\end{equation}
In type II limits, $\mathfrak{q}^2_{\rm min}\sim\mathcal{Q}_{\rm min}$, while $\mathcal{Q}_{\rm min}\succ\mathfrak{q}^2_{\rm min}$ in limits of type III and IV.
Furthermore, using the shorthand notation $s^1={\rm Im}(t^1)$, the geodesic distance $\triangle$ on moduli space is computed to be
\begin{equation}
    \triangle=\int_1^\lambda\,{\rm d}s^1\,\sqrt{\frac{1}{2}\frac{{\rm d}^2K}{{\rm d}(s^1)^2}}=\sqrt{\frac{d}{2}}\log\lambda\,.
\end{equation}
From~\eqref{eq:growth-thm-single} it then follows that
\begin{equation}
    \mathfrak{q}_{\rm min}\sim (s^1)^{(\ell-3)/2}=\exp\left(-\sqrt{\frac{d}{2}}\triangle\right)\equiv\exp\left(-\alpha_\mathtt{A}\triangle\right)\,,
\end{equation}
where $\alpha_{\rm II}=1/\sqrt{2}$, $\alpha_{\rm III}=1$ and $\alpha_{\rm IV}=\sqrt{3/2}$.  Section~\ref{ssec:ESC} interprets this behaviour in light of the Emergent String Conjecture~\cite{Lee:2019wij}. We stress that all single-parameter infinite distance limits in every 4d $\cN=2$ vector multiplet moduli space automatically satisfy the integer scaling relations of~\cite{Lanza:2021udy} recently analysed more generally in~\cite{Grieco:2025bjy,Grieco:2026yip,Etheredge:2026jun}.

\subsection{Multi-parameter limits of simple 
normal 
crossing}\label{ssec:moderatelynormalcrossing}
Moving on to multi-parameter limits, we are faced with the important question of the local model of such a singularity. While for single-parameter limits it is clear that the VPHS close to a singularity undergoes a degeneration over a punctured disk $\Delta^\ast$, in the multi-parameter case the local model heavily depends on the singularity structure of (the components of) the boundary divisor of a compactification of $M_n$. While there is a precise answer in the geometric setting~\cite{viehweg}, the situation is, a priori, less clear from the pure supergravity perspective. However, precise statements can be made for those limits along which the vector multiplet moduli space $M_n$ can be resolved to a simple normal crossing form:\footnote{This assumption also underlies the analysis of~\cite{CKS,complexVHS-multi}.} 

\begin{definition}
    {\rm (Simple\,\,normal\,\,crossing)}\\ A multi-parameter limit ${\rm Im}(t^i)\to\infty$, $i=1,\dots,s$, is said to be {\rm of\,\,simple\,\,normal\,\,crossing\,\,form} if the corresponding local model is $\left(\Delta^\ast\right)^s\times\Delta^t$ with $s+t=n$. In other words, for local coordinates $z^i$ defined as in~\eqref{eq:covering-coords}, the singular locus of $M_n$ is given by $z^1\dots z^s=0$.
\end{definition}

For now we will restrict ourselves to multi-parameter limits of simple normal crossing type and we will come back to the physical significance of this condition in more detail in Section~\ref{ssec:SNC}. The limit ${\rm Im}(t^i)\to\infty$ for $i=1,\dots,s$ is characterised by a nilpotent log-monodromy matrix $N=N_1+\dots N_s$, which is subject to the same constraints as in the single-parameter case. Thus, also multi-parameter limits can be classified into (primary singularity) types I, \dots, IV. Moreover, the simple normal crossing property implies that the various log-monodromy matrices $N_i$ commute, $[N_i,N_j]=0$ for $i\neq j$.

To study the scaling of the physical couplings near higher codimension infinite distance boundaries, we use the multi-parameter generalisation of~\eqref{eq:nilpotent-single} and~\eqref{eq:growth-thm-single} established for quasi-unipotent monodromy in~\cite{CKS} and for more general complex VPHS in~\cite{complexVHS-multi}. 
More precisely, near the infinite distance degeneration $t^i\to\ii\infty$ with $i=1,\dots,s$, the period vector $\Pi$ can be approximated as
\begin{equation}\label{eq:nilpotent} \Pi=\exp\left[\sum_{i=1}^sN_it^i\right]\left(\mathbf{a}_0+\sum_{r_i\geq0}\mathbf{a}_{r_1\dots r_s}e^{2\pi\ii r_it^i}\right)\,,
\end{equation}
where the coefficients $\mathbf{a}$ depend on the coordinates $t^j$, $j=s+1,\dots,n$, which are not scaled to infinity in the limit and the sum is taken over all non-zero multi-indices of length $s$. As in the single-parameter case it follows that the scalars ${\rm Re}(t^i)$ enjoy axionic shift symmetries. The scaling of the Hodge norm of $q\in{\rm Gr}_{\ell_1}\cap...\cap{\rm Gr}_{\ell_s}$ in the limit $t^i\to\ii\infty$, $i=1,\dots,s$, is given by
\begin{equation}\label{eq:growth-thm}
    \|q\|^2\sim\prod_{i=1}^s\left(\frac{{\rm Im}(t^{k_i})}{{\rm Im}(t^{k_{i+1}})}\right)^{\ell_i-3}\,,
\end{equation}
where the ordering $k_1,\dots,k_s$ depends on the chosen growth sector,
\begin{equation}
    \cR_{1\dots s}=\bigg\{\,t^i\,\,\bigg\vert\,\frac{{\rm Im}(t^{k_1})}{{\rm Im}(t^{k_2})},\frac{{\rm Im}(t^{k_2})}{{\rm Im}(t^{k_3})},\dots,\frac{{\rm Im}(t^{k_{s-1}})}{{\rm Im}(t^{k_s})},{\rm Im}(t^{k_s})>\gamma>1\bigg\}\,.
\end{equation}
Using the local model $(\Delta^\ast)^s\times\Delta^t$ of simple normal crossing boundaries it was established in~\cite{CKS,complexVHS-multi} that there is a Deligne splitting~\eqref{eq:Deligne} for all infinite distance limits of this type and therefore also a notion of graded spaces as in~\eqref{eq:graded}. As before, the explicit construction of the spaces $I^{p,q}$ implies \eqref{eq:LimGr} and \eqref{eq:lmin}. From these expressions one can again derive the scaling behaviour of the couplings and show that the physical 1- and 2-forms defined in~\eqref{eq:1-form} and~\eqref{eq:2-form} satisfy the parametric hierarchy~\eqref{eq:hierarchy-couplings}. Furthermore, in limits where all parameters scale at the same rate $\lambda$ (so-called EFT string limits \cite{Lanza:2021udy}),
\begin{equation}\label{eq:homo-scaling}
  {\rm Im}(t^i) \sim e^i  \,   \lambda\, \to \infty \,,
\end{equation}
the scalings  \eqref{eq:distance-scaling} can be established. More general simple normal crossing multi-parameter limits have been analysed in detail in~\cite{Monnee:2025ynn} in the context of Calabi--Yau compactifications; these results likewise carry over immediately to the supergravity setting, which in particular includes the lower bound $\alpha\geq \frac{1}{\sqrt{2}}$ on the exponential rate of $\mathfrak{q}_{\rm min}$.

\subsection{A comment on instanton corrections}\label{ssec:instantons}
The real VPHS underlying the vector multiplet moduli space of a 4d $\cN=2$ supergravity is a consequence of classical supersymmetry. Imposing supersymmetry also at the level of the quantised theory likewise gives rise to a VPHS. However, as shown in a simple example in~\cite{Kaufmann:2026fli,Kaufmann:2026mha} (based on results of~\cite{Camara:2008zk,Berg:2005ja}), it is in general not true that this quantum VPHS coincides with the classical VPHS. Even if the fully quantum corrected VPHS may in general not be known, the results of the previous sections guarantee that it obeys the properties in Claim \ref{Claim1} at infinity.\footnote{In particular, there must exist flat coordinates in which the expected infinite distance behaviour is manifest. This is a notable difference compared to 4d $\cN=1$ theories \cite{Kaufmann:2026fli,Kaufmann:2026mha,Kaufmann:2026tsy}.} 

Whether quantum corrections are relevant near a given infinite distance regime in moduli space is of course hard (if not impossible) to answer in general from a pure supergravity perspective. Nevertheless, the existence of certain non-perturbative corrections to the K\"ahler potential and prepotential are already encoded in the classical Hodge theory on $M_n$.  As was shown in~\cite{Bastian:2021eom}, near all infinite distance limits with $s=n$ in previous notation (except those of type IV$_n$), the metric derived from the K\"ahler potential~\eqref{eq:Hodge-Kpot} becomes singular. This can be traced back to the fact that in these limits the set of derivatives $\{D_i\Pi\}$ of the period vector fails to be linearly independent, i.e. $\dim\left({\rm span}_\mathbb{C}\{D_i\Pi\}\right)<n$. The only way to cure this singularity in the metric is to include the leading (but exponentially suppressed) terms in the nilpotent orbit expansion~\eqref{eq:nilpotent} of the period vector.\footnote{We stress again that while the analysis of~\cite{Bastian:2021eom} is phrased in terms of Calabi--Yau geometry, the presence of these singularities is a statement solely about the underlying VPHS which is already available at the level of the supergravity.  Indeed, as was exemplified for small $n$ in~\cite{Bastian:2021eom}, the period vector $\Pi$ in~\eqref{eq:period-vector} is completely specified (in the nilpotent or $\mathfrak{sl}(2)$-orbit approximation) by the type of infinite distance degeneration with no further reference to any underlying Calabi--Yau geometry.\label{fn:CYperiods}} Put differently, also exponential contributions to the prepotential $\mathbb{z}$ are captured by the supergravity. As pointed out in~\cite{Bastian:2021eom}, even though one can infer a contribution to the prepotential of the form  
\begin{equation}
    \mathbb{z}_{\rm n.p.}=\sum_{k\geq0}n_k e^{2\pi\ii\alpha_{i,k}t^i}\,,
\end{equation}
a clear interpretation of it is not readily available from the supergravity perspective.\footnote{The need for quantum corrections in the present setup was also pointed out in~\cite{Cecotti:2020rjq}. Whereas this reference argues for a form of ``completeness'' of non-perturbative corrections based on the therein proposed Structural Criterion, \cite{Bastian:2021eom} deduces the existence of corrections exclusively (as emphasised in Footnote~\ref{fn:CYperiods}) based on the underlying VPHS with no further UV/ quantum gravity input.}

\subsection{Compatibility with the Emergent String Conjecture}\label{ssec:ESC}
Based on the VPHS underlying the vector multiplet moduli space of any 4d $\cN=2$ supergravity theory, the scalings \eqref{eq:distance-scaling} and \eqref{eq:hierarchy-couplings} of the physical 1- and 2-form couplings have been derived in all infinite distance limits of simple normal crossing type in $M_n$. An analogous result was established for 5d ${\cal N}$=1 supergravity theories in~\cite{Kaufmann:2024gqo,david-vicente}.\footnote{Ref. \cite{Kaufmann:2024gqo} proves this with the help of a certain completeness hypothesis regarding the spectrum of supergravity strings \cite{Katz:2020ewz}, while  \cite{david-vicente} shows that the result holds already in supergravity without further assumptions.} In this section we discuss how these scalings are related to the Emergent String Conjecture \cite{Lee:2019wij}. Much of the following discussion closely follows  analogous observations in the context of Calabi--Yau compactifications, in particular the discussion in \cite{Monnee:2025ynn}, which we generalise to the present setting and briefly summarise for completeness.

\paragraph{Single-parameter limits.} 
By comparison with \eqref{scaling-gen}, the exponential rates $\alpha_{\mathtt{A}}$ appearing in the scaling of the gauge couplings~\eqref{eq:distance-scaling} match precisely the values expected for asymptotically vanishing gauge couplings in emergent string limits ($\mathtt{A}=$II), decompactification limits to 6d ($\mathtt{A}=$III) and decompactification limits to 5d ($\mathtt{A}=$IV)~\cite{Agmon:2022thq,Etheredge:2022opl}. In limits of type II$_b$, the fastest vanishing gauge couplings are parametrised by a basis $q^i_{\rm II} \in {\rm Gr}_2$, $i = 1, \ldots 2 +b$, satisfying $\|q^i_{\rm II}\|\sim\mathfrak{q}_{\rm min}$. They correspond to the perturbative parts of the gauge group which become weakly coupled in an asymptotic weak coupling limit of a dual fundamental string theory (the emergent string of the limit). In limits of type III, $q^i_{\rm III} \in {\rm Gr}_1$, $i=1,2$, give rise to $\|q^i_{\rm III}\|\sim\mathfrak{q}_{\rm min}$ which are interpreted as the two Kaluza-Klein (KK) gauge couplings for decompactification along a torus to 6d, while in limits of type IV, $\mathfrak q_{\rm min}$ corresponds to the single generator $ q_{\rm IV} \in {\rm Gr}_0$, which can be identified as the KK gauge coupling for a decompactification to 5d.

Furthermore, the couplings also determine the masses and tensions of suitable BPS objects which, \emph{if} they exist in the theory, identify the nature of the limit. First, as explained in~\cite{Lanza:2021udy}, the minimal physical 2-form coupling sets the tension of an EFT string inducing the limit,
\begin{equation}\label{eq:EFT-tension}
    \frac{\mathcal{Q}_{\rm min}}{M_{\rm Pl}^2}\sim \frac{T_{\rm EFT}}{M_{\rm Pl}^2}=-e^i\frac{\partial K}{\partial s^i}\,.
\end{equation}
Here $e=(e^i)$ is the vector introduced around~\eqref{eq:covering-coords}, now interpreted as the magnetic charge vector of the EFT string, in particular $e\in\mathbb{Z}^{n}$. The EFT string is a solitonic solution of the supergravity which becomes tensionless in the limit. The backreaction on the moduli fields close to its core induces the infinite distance limit. Second, the minimal physical 1-form coupling corresponds to the mass of a (so far hypothetical) BPS particle,
\begin{equation}\label{eq:BPS-mass}
    \frac{m_{\rm BPS}}{M_{\rm Pl}}\sim\mathfrak{q}_{\rm min}\,.
\end{equation}
Of course, the existence of these BPS objects is by no means guaranteed from the supergravity perspective alone and, in fact, represents the actual non-trivial content of the Distance and Emergent String Conjecture (at least in the context of 4d ${\cal N}=2$ theories due to the analysis of the previous sections). For our purposes it is ensured only by an appropriate version of the completeness hypothesis~\cite{Polchinski:2003bq,Banks:2010zn}. This completeness assumption necessitates an \emph{integral} structure underlying the VPHS on $H_\mathbb{C}\to M_n$ as we will discuss in the next subsection.

However, the key point is the following: If we assume the existence of an EFT string of tension $T_{\rm EFT}$ and of a \emph{tower} of BPS particles of mass scale $m_{\rm BPS}$, then the scaling behaviour of $T_{\rm EFT}$ and $m_{\rm BPS}$, as established from the supergravity analysis alone, is precisely as required 
for the interpretation of type II limits as emergent string limits and type III and type IV limits as decompactification limits to six respectively five dimensions. This interpretation is corroborated by the following observations:
\begin{itemize}
    \item In a single-parameter type II limit $t^1\to\ii\infty$ there is only a single axion ${\rm Re}(t^1)$ which can be dualised into a 2-form to which an EFT string couples. This implies in particular that the EFT string whose tension scales as $\cQ_{\rm min}$ is unique and can play the role of the unique fundamental string becoming tensionless in the limit, as required by the Emergent String Conjecture. The relation $\cQ_{\rm min} \sim \mathfrak{q}^2_{\rm min}$ furthermore implies $T_{\rm EFT} \sim m^2_{\rm BPS}$, in agreement with the interpretation of the BPS particle tower as an accomanying KK or winding tower of the tensionless string.
    \item In limits of type III, the lightest states have charge vectors in the two-dimensional space ${\rm Gr}_1$ (recall the discussion around~\eqref{eq:graded}), thereby hosting the putative two KK towers at mass scale $m_{\rm BPS}$ in a decompactification limit to 6d. Similarly, the lightest states in a type IV limit have charge vectors lying in the 1-dimensional space ${\rm Gr}_0$.
\end{itemize}
In string theoretic constructions of 4d $\cN=2$ supergravities, the existence of the EFT strings and towers, the fact that in type II limits the EFT string is a critical string, and the identification of the BPS particle towers as KK/winding towers follow from the details of the geometry, as shown in detail for Type II Calabi--Yau compactifications in \cite{Lee:2019wij,Hassfeld:2025uoy,Monnee:2025ynn}.

\paragraph{Multi-parameter limits.}
This discussion carries over to multi-parameter limits under the assumption of simple normal crossing. At first sight, it therefore seems that all infinite distance limits in $M_n$ are consistent with the Emergent String Conjecture. However, there is a subtlety related to multiple single-parameter type II limits which do not intersect any type III or IV divisor and also at their intersection do not enhance to type III or IV. For these limits to be consistent, the emergent string has (in particular) to be unique \cite{Lee:2019wij}. At the level of the supergravity, this uniqueness amounts to the uniqueness of the charge vector $e$ realising the scaling $\mathcal{Q}^2_e\sim\mathcal{Q}^2_{\rm min}$. While in single-parameter type II limits uniqueness follows automatically, in the multi-parameter case there are multiple EFT string charge vectors $e_j=(\delta^i_j)$, $j=1,\dots,s$, for which all $\cQ_{e_j}^2$ sit at the same scale $\cQ_{\rm min}^2$.\footnote{The analogues of type II limits for 5d $\cN=1$ supergravities were termed tensor limits in~\cite{Kaufmann:2024gqo} and uniqueness of the leading physical 2-form coupling was proven in all such limits. More precisely, in the language of~\cite{Kaufmann:2024gqo}, multi-parameter limits are limits with $|{\cal J}_\lambda| > 1$, which are always vector limits (i.e. limits in which all ${\cal Q }\succ {\mathfrak q}_{\rm min}^2$). This means that the analogue of an enhancement ${\rm II} + {\rm II} \rightarrow {\rm II}$ (analysed in detail in~\cite{cgm}) does not exist in 5d. This is consistent with the above observation because the corresponding enhancements in 4d are not connected to a type III or IV degeneration.}

In fact, this problem also arises in string theoretic constructions of 4d $\cN=2$ supergravities as Calabi-Yau compactifications~\cite{Hassfeld:2025uoy}. To see in what sense uniqueness of the charge vector $e$ realising scaling of $\cQ_{\rm min}$ is guaranteed also in multi-parameter type II$_b$ limits, we follow the argument of~\cite{Hassfeld:2025uoy} and show that all physical 1-form couplings $\mathfrak{q}^2_i\sim\mathfrak{q}^2_{\rm min}\sim\cQ_{\rm min}$, $i=1,\dots,b+2$, are mutually electric. For notational simplicity we consider only a subset of two charge vectors $q_i\in{\rm Gr}_{2,i}$, $i=1,2$, realising $\mathfrak{q}^2_i\sim\mathfrak{q}^2_{\rm min}$ via~\eqref{eq:1-form}. Due to the simple normal crossing property of the type II$_b$ limit, the log-monodromy matrices $N_1$ and $N_2$ satisfy
\begin{equation}
    N_1^2=N_2^2=0\,,\quad (N_1+N_2)^2=2N_1N_2=0\, \label{eq:N1N2cond}
\end{equation}
since $[N_1,N_2]=0$. Even without quasi-unipotent monodromy it is true that $N_{1,2}$ induce linear isomorphisms $N_j:{\rm Gr}_{4,j}\to{\rm Gr}_{2,j}$~\cite{complexVHS} and hence
\begin{equation}
    (q_1,q_2)=(N_1\bar q_1,N_2\bar q_2)=-(\bar q_1,N_1N_2\bar q_2)=0\,,
\end{equation}
where $(\cdot,\cdot)$ denotes the symplectic pairing introduced in Section~\ref{ssec:global}.\footnote{The second equality comes from the fact that $(\cdot,\cdot)$ is monodromy invariant. Furthermore, as shown in \S61 of~\cite{complexVHS}, the pairing $(\cdot,\cdot)$ is actually invariant under the semisimple part of the monodromy. Writing $T_u=e^N$ and using that $(\cdot,\cdot)$ is conjugate linear in the second argument yields the second equality above.}
In light of the discussion around~\eqref{eq:magnetic-G}, the pairing $(\cdot,\cdot)$ can be interpreted as a Dirac pairing and so the vanishing of $(q_1,q_2)$ indeed shows that the corresponding physical 1-form couplings are mutually electric. In other words, the corresponding gauge theories ${\rm U}(1)_1$ and ${\rm U}(1)_2$ are mutually local. This is a pure supergravity statement.

As in~\cite{Hassfeld:2025uoy}, to deduce some form of uniqueness of the charge vector $e$ (and hence the putative emergent string), a pure supergravity analysis is not sufficient and further quantum gravity input (related to the Quantum Gravity Principles~\ref{Swamp1} and~\ref{Swamp2} below) is needed. Concretely, we assume that all involved single-parameter type II limits are indeed emergent string limits. For the situation described above this means that there are fundamental strings with charge vectors $e_{1,2}$ (which realise $\cQ_{1,2}\sim\cQ_{\rm min}$) as well as BPS particle states charged under the dual perturbative gauge groups $G_1$ and $G_2$ in the respective single-parameter limits.\footnote{The gauge couplings of the perturbative gauge groups $G_{1}$ and $G_2$ in the stringy dual frame are of course identified with $\mathfrak{q}_{1,2}\sim\mathfrak{q}_{\rm min}$.} The two duality frames can be compatible only if the associated gauge theories with gauge groups $G_{1}$ and $G_2$ are mutually local in the above sense. The vanishing of the Dirac pairing $(q_1,q_2)$ discussed before is therefore a necessary condition for compatibility of multi-parameter type II limits with the Emergent String Conjecture. Note that the (simple) normal crossing condition, which is responsible for \eqref{eq:N1N2cond}, was key to show mutual compatibility of the two duality frames. This, in fact, holds much more generally, as we will discuss now.

\section{The 4d \texorpdfstring{$\cN=2$}{N=2} vector multiplet Swampland}\label{sec:swamp}
The observations of the previous section raise the question to what extent the characteristic predictions of the Emergent String Conjecture are an automatic consequence of supersymmetry alone rather than a non-trivial quantum gravity constraint. While for the single-parameter limits all necessary information needed to show Claim \ref{Claim1} is indeed available from the geometry of the supergravity moduli space $M_n$ alone without further input, the analysis of multi-parameter limits is based on the restriction to simple normal crossing degenerations. This restriction is, in fact, a non-trivial one; we will illustrate this in Example~\ref{ex:non-snc} below. So far, the simple normal crossing condition has appeared as a {\it sufficient} technical condition for the results of Section \ref{sec:Asymptotics} to follow from the VPHS also for multi-parameter limits~\cite{Deng2022OnTN,complexVHS-multi}. In this section we will propose that it is, in fact, also {\it necessary} for consistency with the Distance and Emergent String Conjecture. 

In fact, we will put forward two quantum gravity principles for a 4d $\cN=2$ supergravity theory to be consistent with quantum gravity:

\begin{swamp} \label{Swamp1}
    The vector multiplet moduli space of a 4d $\cN=2$ supergravity theory consistent with quantum gravity must admit an integral variation of polarised Hodge structure.
\end{swamp}

As we will discuss in Section \ref{ssec:completeness}, this principle is a consequence of the Completeness Conjecture \cite{Polchinski:2003bq} and of the No-Global-Symmetries Conjecture \cite{Banks:1988yz,Banks:2010zn}. It is furthermore required to ensure consistency with the Distance and Emergent String Conjecture because it is a necessary condition for the existence of asymptotic towers of states. That an underlying integral structure is related to quantum gravitational principles was already anticipated in~\cite{Cecotti:1989kn}.

\begin{swamp} \label{Swamp2}
    The vector multiplet moduli space $M_n$ of a 4d $\cN=2$ supergravity coupled to $n$ vector multiplets consistent with quantum gravity can be compactified to a variety $\overline{M}_n$ such that $\overline{M}_n\setminus M_n$ has only simple normal crossing singularities. 
\end{swamp}

This principle will be seen, in Section \ref{ssec:SNC}, to be necessary (possibly up to the condition of simplicity) to guarantee a well-defined duality frame in multi-parameter infinite distance limits and hence required for the Distance and Emergent String Conjecture to hold.

Note that both of these principles are necessary (but not sufficient) conditions for the Structural Criterion proposed in~\cite{Cecotti:2020rjq}. The latter proposes---based on a result of~\cite{Green2010}---that consistent projective special K\"ahler manifolds are those vector multiplet moduli spaces with a geometric origin.\footnote{To decide whether a given (integral) VPHS has a geometric origin is a long-standing open problem in mathematics~\cite{GROTHENDIECK1969299,Kreutz}.}

\subsection{Codimension one: Integrality and completeness}\label{ssec:completeness}
Around~\eqref{eq:magnetic-G} we have recalled that the structure group ${\rm Sp}(2n+2,\mathbb{R})$ of the vector bundle $H\to M_n$ is interpreted as the electric-magnetic duality group of the supergravity. The fact that this duality group has real coefficients is a consequence of the projective special K\"ahler geometry of the vector multiplet moduli space $M_n$ as it ensures the existence of a real vector bundle $H\to M_n$. For a given 4d $\cN=2$ supergravity to lie in the Landscape, it is strongly believed that this theory has to satisfy a certain form of the Completeness Hypothesis~\cite{Polchinski:2003bq}: a (sub)lattice of charges is populated by physical states. This includes in particular charged (BPS) particle states which then transform under the electric-magnetic duality group ${\rm Sp}(2n+2,\mathbb{R})$ of the supergravity. Naively, this would allow for particles with real (and in particular irrational) electric/magnetic charges. However, as was argued in~\cite{Banks:2010zn}, having two particles with mutually irrational charges leads to a violation of the covariant entropy bound of~\cite{Fischler:1998st,Bousso:1999xy} and furthermore gives rise to a global symmetry, which is widely believed to be forbidden in quantum gravity~\cite{Banks:1988yz,Banks:2010zn}. Thus, the only consistent way completeness can be realised is via a spectrum of rational charges. A rescaling of this rational charge lattice then leads to an integral lattice of populated charges. In other words, consistency of the supergravity with the Completeness Hypothesis and the absence of global symmetries demands that at each point $t\in M_n$ in moduli space there exists a $\mathbb{Z}$-module $H_{\mathbb{Z},t}$ such that the fiber $H_t$ of the real vector bundle $H\to M_n$ can be written as $H_t=H_{\mathbb{Z},t}\otimes_\mathbb{Z}\mathbb{R}$. This integral structure can then be interpreted as the electric-magnetic charge lattice, fully populated by the completeness hypothesis, which is acted upon by a fixed subgroup ${\rm Sp}(2n+2,\mathbb{Z})\subset{\rm Sp}(2n+2,\mathbb{R})$. 

This integral structure $H_\mathbb{Z}$ underlying the real vector bundle $H\to M_n$ and its complexification $H_\mathbb{C}\to M_n$ is also interesting from a technical perspective as it ensures quasi-unipotent monodromies. Whereas the analysis in Section \ref{sec:Asymptotics} relies on the results of~\cite{Deng2022OnTN,complexVHS,complexVHS-multi} which hold for an arbitrary real (even complex) VPHS, the original works~\cite{schmid,CKS} are formulated in the more restrictive setting of a VPHS with quasi-unipotent monodromy. Of course, for VPHS with a geometric origin, such an integral structure $H_\mathbb{Z}$ is always available and simply given by the middle cohomology with integer coefficients of the underlying space. For a general supergravity this is  in general not true and only ensured via the Completeness Hypothesis.

\subsection{Higher codimension: 
Normal
crossings and dualities}\label{ssec:SNC}
To connect the simple normal crossing condition (arising as a sufficient condition for the results of ~\cite{Deng2022OnTN,complexVHS-multi})  to a physical and quantum gravitational statement, we make use of the idea behind the Distance and Emergent String Conjectures. Namely, near every infinite distance boundary there should exist a dual formulation of the original theory. In more detail, we show

\begin{claim}\label{claim:NC}
    The existence of a well-defined weakly coupled gravitational  duality frame at a multi-parameter infinite distance degeneration in $M_n$ requires that the involved single-parameter divisors are at (not necessarily simple) normal crossing.
\end{claim}

Note that the existence of such a gravitational duality frame is not in contradiction with the presence of (possibly mutually non-local) rigid field theory sectors~\cite{Marchesano:2023thx,Marchesano:2024tod,Castellano:2024gwi,Blanco:2025qom,Castellano:2026bnx,Aoufia:2026mqb,Monnee:2025ynn, Hattab:2025aok,Hattab:2026lho,cgm}.

\begin{proof}[Argument]
Let us focus on the case of two single-parameter infinite distance divisors $D_1$ and $D_2$. Along each divisor we assume the existence of a well-defined and weakly coupled dual description of the theory, which we denote by ${\rm DF}_{1,2}$. To be more precise, we choose a local coordinate $z_1$ on $M_n$ such that $D_1$ is given by $D_1=\{z_1=0\}$. Then, by assumption, ${\rm DF}_1$ satisfies the following (related) properties:
\begin{enumerate}
    \item A suitable inverse power of the coordinate ${\rm Im}(t^1)$ (where $z_1=\exp(2\pi\ii t^1)$) has the dual interpretation of a small coupling parameter organising the perturbation series of the theory in the dual frame ${\rm DF}_1$. This means that each physical quantity has a clean expansion into a classical, loop and non-perturbative terms in this coupling.
    \item Along $D_1$, the spectrum of light dynamical states is well-defined. This means that each state falls uniquely into one of the three categories: light and dynamical, heavy and integrated out, non-perturbative and classical.\footnote{This point is also emphasised in recent literature on the M-theoretic Emergence Proposal, see e.g.~\cite{Blumenhagen:2024lmo,Artime:2026kfq}.} 
    \item From the single-parameter nilpotent orbit expansion~\eqref{eq:nilpotent-single} it follows that the covering coordinate $2\pi\ii t^1=\log(z_1)$ has axionic real part. The axionic shift induces a monodromy action $T_1$ upon encircling the divisor $D_1$ as $z_1\to e^{2\pi\ii}z_1$. 
\end{enumerate}
Of course, the same properties hold for the analogously defined dual frame ${\rm DF}_2$ at generic points of $D_2$ with local coordinate $z_2$.

The key observation is that in case the second divisor $D_2$ meets $D_1$ in a point of tangency of order $k\geq1$, this tangency introduces a relation between the local coordinates $z_1$ and $z_2$ near the intersection making it impossible to consistently recover the single-parameter dual frames ${\rm DF}_i$ starting from the intersection $D_1\cap D_2$.\footnote{In full generality, it is not clear whether the intersection of two infinite distance divisor components sits itself at infinite distance in moduli space~\cite{Grimm:2018ohb}. For Calabi--Yau compactifications of Type IIB string theory, this statement is proven in~\cite{geronimo}.} The only time such a relation is absent is in case $D_1$ and $D_2$ meet transversally, i.e. if $D_1$ and $D_2$ are at normal crossing. As a result, in absence of normal crossing we can find a physical inconsistency to all three properties at the intersection $D_1\cap D_2$: 
\begin{enumerate}
    \item Along both divisors, physical quantities can be organised into a classical piece, perturbative corrections and non-perturbative contributions. At the intersection $D_1\cap D_2$, there should therefore exist a duality frame in which there are two independent perturbative series, one in ${\rm Im}(t^1)^{-1}$ and one in ${\rm Im}(t^2)^{-1}$,
    \begin{equation}
        \cA({\rm Im}(t^1),{\rm Im}(t^2))=\sum_{n,m\geq0}\cA_{n,m}{\rm Im}(t^1)^{-n}{\rm Im}(t^2)^{-m}+\,\text{non-pert}\,,
    \end{equation}
    for every physical quantity $\cA$. Only in this way is it possible to ensure that one recovers the perturbation series valid away from the intersection along generic points of $D_1$ and $D_2$. We now assume that $D_1$ and $D_2$ meet in a point of tangency of order $k\geq1$, i.e. we can write $z_1=z_2^k$ near the intersection $D_1\cap D_2$.\footnote{Notice that $z_1=z_2$ does not define a normal crossing intersection, which would be of the form $z_1z_2=0$.} It follows that the double expansion of $\cA$ collapses into
    \begin{equation}
        \cA({\rm Im}(t^1),{\rm Im}(t^2))=\sum_{n,m}\tilde{\cA}_{n,m}{\rm Im}(t^2)^{-n-m}+\,\text{non-pert}=\sum_{p}\cA_p{\rm Im}(t^2)^{-p}+\,\text{non-pert}\,,
    \end{equation}
    and it is not possible to uniquely identify the $\tilde{\cA}_{n,m}$ from the $\cA_p$. In other words, the $D_1$ loop-order cannot be recovered from the perturbative expansion at tangential meeting locus $D_1\cap D_2$.
    \item At generic points of both divisors, the spectrum of states in the respective duality frame ${\rm DF}_i$ is organised as
    \begin{equation}\label{eq:singlespectrum}
        \cS_i:\,\,\underbrace{\left[m_\ell\sim \frac{1}{{\rm Im}(t^i)^{\ell_i}}\,,\,\,\ell_i>0\right]}_{\text{light}}\oplus\underbrace{\left[m_h\sim \frac{1}{{\rm Im}(t^i)^{h_i}}\,,\,\,h_i\leq0\right]}_{\text{heavy}}\oplus\underbrace{\left[e^{-S_i}\sim e^{\ii t^i}\sim|z_i|\right]}_{\text{non-perturbative}}\,.
    \end{equation}
    We stress that the spectrum $\cS_i$ focuses on the gravitationally coupled states and so does not include potential rigid field theory sectors whose masses are not polynomial in ${\rm Im}(t^i)$. A consistent spectrum $\cS_{12}$ at the intersection $D_1\cap D_2$ must again allow to recover the single-parameter spectra $\cS_i$ when moving (along one divisor) away from the intersection. In general, $\cS_{12}$ will take the form $\cS_1\otimes\cS_2$ and all possible combinations of states in~\eqref{eq:singlespectrum} will appear. The precise splitting into light and heavy states depends on the growth sector, see Point 3. below. Among the light states, there will be states with mass scaling as 
    \begin{equation}
        \cS_{12,\,{\rm light}}\supset \left[m_{\ell\ell}\sim 
        \frac{1}{{\rm Im}(t^1)^{\ell_1}{\rm Im}(t^2)^{\ell_2}}\,,\,\,\ell_i>0\right]\,.
    \end{equation}
    For a point of tangency, we again write $z_1=z_2^k$, $k\geq1$, and find that
    \begin{equation}
        m_{\ell\ell}\sim \frac{1}{{\rm Im}(t^2)^{\ell_1+\ell_2}}\prec \frac{1}{{\rm Im}(t^2)^{\ell_2}}\,,\quad m_{\ell\ell}\sim \frac{1}{{\rm Im}(t^1)^{\ell_1+\ell_2}}\prec \frac{1}{{\rm Im}(t^1)^{\ell_1}}\,,
    \end{equation}
    showing that we do not recover the scaling of the single-parameter duality frames when moving away from the intersection.\footnote{Of course, this scaling argument can also be applied to the full Distance Conjecture tower. We stress that we do not need the tower for this argument as the inconsistency already appears for a finite number of states.} On top of this inconsistency, a similar argument can be made for the non-perturbative sector which consists of all sectors in $\cS_1\otimes\cS_2$ including at least one of the single-parameter non-perturbative sectors. The tangency between $D_1$ and $D_2$ then shows that there are (at least) $k-1$ ${\rm DF}_2$-instantons below the lightest ${\rm DF}_1$-instanton, thereby blurring the distinction between perturbative and non-perturbative states in $\cS_1$,
    \begin{equation}
        m_{{\rm inst},2}\sim |e^{\ii t^2}|=|e^{\frac{\ii}{k}t^1}|\precsim |e^{\ii t^1}|\sim m_{{\rm inst},1}\,.
    \end{equation}
    Notice that even for $k=1$ the spectrum does not coincide with $\cS_1$ as there are now additional states at the scale $m_{{\rm inst},1}$.
    \item The axionic shift symmetries of $t^{1,2}$ near $D_{1,2}$ induce monodromy actions $T_{1,2}$ on $\Pi$. In the single-parameter limits, the unipotent parts of the $T_i$ define the log-monodromy matrices $N_i$ via $T_{u,i}=\exp N_i$.\footnote{Recall the lack of quasi-unipotency of the $T_i$ as discussed in Section~\ref{ssec:single}. The results of~\cite{complexVHS,complexVHS-multi} show that the respective limiting mixed Hodge structures are nevertheless characterised by the $N_i$.} In the case of tangency, the log-monodromy matrices do not commute, $[N_1,N_2]\neq0$. This means that there are two different nilpotent orbit expansions of the periods near $D_1\cap D_2$,
    \begin{equation}
        \Pi_{12}=e^{t^1N_1}e^{t^2N_2}\mathbf{a}_0\,,\quad \Pi_{21}=e^{t^2N_2}e^{t^1N_1}\mathbf{a}'_0\,.
    \end{equation}
    Correspondingly, there is no well-defined growth theorem telling us about the light spectrum at the intersection $D_1\cap D_2$ and, as before, we conclude that in the case of a tangency between $D_1$ and $D_2$, there is no well-defined loop-grading and no well-defined organisation $\cS_{12}$ of the spectrum of states.
\end{enumerate}
\end{proof}

The order of tangency $k$ between $D_1$ and $D_2$ can in principle be infinite. An instance of this is presented below in Example~\ref{ex:non-snc}. For finite contact order $k\in\mathbb{N}$, the tangency between $D_1$ and $D_2$ can be removed into a transverse intersection via a sequence of $k$ blow-ups. The endpoint of this sequence then gives rise to the physically relevant and well-behaved compactification $\overline{M}_n$ of $M_n$. In particular, there exists a normal crossing compactification of moduli space and the theory does not necessarily reside in the Swampland. For infinite order of tangency, such a resolution process does not work. In this case, the theory is indeed inconsistent with the Distance and Emergent String Conjectures and as such believed to lie in the Swampland.

That the normal crossing condition on $M_n$ is in fact a relevant Swampland condition is further substantiated by Example 32 of~\cite{haydys2020} which presents a 4d $\cN=2$ field theory vector multiplet moduli space for which this condition is not satisfied.

Notice that we phrased the previous discussion in terms of $D_1$ and $D_2$ meeting transversally, i.e. $D_1$ and $D_2$ are at normal crossing. The condition of technical relevance for~\cite{Deng2022OnTN,CKS,complexVHS-multi}, however, is that of \emph{simple} normal crossing. Even if $D_1$ and $D_2$ intersect normally, it could still be that some $D_i$ become(s) non-reduced at the intersection. This means that the local model near $D_1\cap D_2$ would not be $(\Delta^\ast)^2\times \Delta^{n-2}$, but rather include a multi-cover of the first two factors. This is reflected in the monodromy of the vector multiplet scalars $t^{1,2}$ around the singularity $D_1\cap D_2$, i.e. the periodicity of the axions ${\rm Re}(t^{1,2})$, see the discussion after~\eqref{eq:2-form},
\begin{equation}\label{eq:monodromy}
    t^i\to t^i+e^i\quad{\rm as}\quad z_i\to e^{2\pi\ii}z_i\,.
\end{equation}
Whether it is possible to change coordinates so that this monodromy becomes minimal, i.e. $e^i=1$, is again a completeness question. Indeed, an equivalent way of thinking about the monodromy vector $e=(e^i)$ was discussed in Section~\ref{ssec:ESC} where we recalled that $e$ is the magnetic charge vector of a solitonic string solution coupled to the supergravity. For the infinite distance limits we are interested in, these solutions have been coined EFT strings in~\cite{Lanza:2021udy}. Hence, whether a rescaling to $e^i=1$ is possible depends on the EFT string spectrum and therefore on the version of the completeness hypothesis one is willing to assume. We conjecture that an EFT string of minimal charge exists and hence that the normal crossing singularity can be made simple normal crossing.

The resulting Quantum Gravity Principle \ref{Swamp2} is in fact identical to the requirement of algebraic compactifiability of 4d $\cN=2$ vector multiplet moduli spaces. The latter was proposed in ~\cite{Delgado:2024skw} as a sufficient condition to ensure the geometric compactifiability criterion argued for as a manifestation of finiteness in quantum gravity. Via our Claim \ref{claim:NC}, we have argued why this algebraic compactifiability criterion is also a necessary condition for consistency with 4d $\cN=2$ quantum gravity.

\begin{example}\label{ex:non-snc}
     {\rm (Infinite order tangency)}\\
     The normal crossing property discussed above does not follow automatically from the projective special K\"ahler geometry of the vector multiplet moduli space. To see this explicitly, consider the two-dimensional moduli space
     \begin{equation}
         M_2=\mathbb{C}^2\setminus\left(\{0\}\cup D_1\cup D_2\right)\,,
     \end{equation}
     where $D_1=\{z_1=0\}$, $D_2=\{z_1=e^{-1/z_2}\}$ in terms of complex coordinates on $\mathbb{C}^2$. To describe the projective special K\"ahler structure on $M_2$, we define the covering coordinates
     \begin{equation}
         t^1=\frac{1}{2\pi\ii}\log(z_1)\,,\quad t^2=\frac{1}{2\pi\ii}\log(z_1-e^{-1/z_2})\,
     \end{equation}
     so that the boundary divisor components $D_i$ lie at infinite distance in the limits $t^i\to\ii\infty$. As prepotential we choose
    \begin{equation}
        \mathbb{z}(Z)=\frac{Z^1(Z^2)^2}{Z^0}\,,
    \end{equation}
    where $Z^0=1$ yields $t^i=Z^i/Z^0$. By direct computation one finds
    that $D_1$ is a boundary of type II and $D_2$ a boundary of type III, i.e. both monodromies are unipotent and the log-monodromy matrices satisfy $N_1^2=0$, $N_2^3=0$ (but $N_1,N_2^2\neq0$). Notice that the two components $D_1$ and $D_2$ meet to infinite tangency order due to the exponential function appearing in the definition of $D_2$. Correspondingly, $D_1\cap D_2$ is an essential singularity for which no sequence of blow-ups can lead to (simple) normal crossing. 
\end{example}
At first sight, it might be surprising to see a cubic prepotential appearing in the example violating the normal crossing property, which is automatic in the case of Calabi--Yau compactifications. Note, however, that a possible compactification $\overline{M}_n$ of moduli space is not part of the supergravity data defining the theory. Indeed, the projective special K\"ahler geometry is a statement only about the moduli space $M_n$ itself and not about possible compactifications. As a result, from a pure supergravity perspective we are free to ``choose'' a compactification (if it exists) and with that a boundary divisor $D=\overline{M}_n\setminus M_n$. In contrast, this boundary data is automatically given in the case of Calabi--Yau compactifications and normal crossing is guaranteed~\cite{viehweg}.\footnote{In the geometric case, $M_n$ is quasi-projective~\cite{viehweg}. Starting with a simple normal crossing (SNC) compactification $(\overline{M}_n,D)$ one might nevertheless continue blowing-up along strata of the boundary divisor $D$ so as to produce an infinity of other SNC compactifications of $M_n$. A priori this introduces new boundary components and hence a dependence of the asymptotic physics on the chosen boundary after all. However, this is not so: the new exceptional boundary component $E$ has log-monodromy matrix $N_E\in\sum_i\mathbb{R}_+N_i$ where the sum runs over all components $D_i$ of $D$ which contain the blow-up center. It is known that the limiting mixed Hodge structure associated to $E$ is the same for any $N_E$ in this cone~\cite{complexVHS-multi,CKS} and so there is no new physics associated with $E$. Under additional smoothness assumptions of the compactifications it is true that any birational transformation of an SNC compactification leaves the physics invariant~\cite{wlodarczyk1999toroidal}.} This additional freedom in the supergravity setting is what allows us to find boundaries which cannot be resolved to normal crossing form. 

Incidentally, the function $e^{-1/z_2}$ describing the boundary behaviour of non-simple normal crossing type is not a tame function. Hence this example also seems to be in tension with the tameness conjecture of quantum gravity \cite{Grimm:2021vpn} (see \cite{Grimm:2025lip} for the connection between tameness and the geometric compactifiability of~\cite{Delgado:2024skw}). However, the relation to tameness appears to be indirect and in fact a consequence of holomorphy and therefore supersymmetry, in the following sense: The reason why $e^{-1/z_2}$ is not tame is because of the periodicity of the complex exponential function rather than because of the infinite order of tangency, which lies at the heart of the violation of the simple normal crossing condition. The infinite vanishing order would just as well be implemented by a real function $e^{-1/x}$, $x \in \mathbb R$, which is tame, but supersymmetry dictates holomorphy. Therefore it appears that tameness as such is not directly the driving force behind Quantum Gravity Principle \ref{Swamp2}. It would be interesting to understand this better.

\subsection{(Locally) symmetric projective special K\"ahler manifolds}
It is informative to apply our findings to moduli spaces of (locally) symmetric type. Locally symmetric spaces are manifolds of the form
\begin{equation}\label{eq:loc-sym}
    M_n=\Gamma\setminus G/K\,,
\end{equation}
where $G$ is a (reductive) Lie group with $K$ its maximal compact subgroup and $\Gamma$ is a discrete subgroup of isometries of $G$. Only a special subclass of locally symmetric spaces are also projective special K\"ahler manifolds. This subclass is constructed from non-compact globally symmetric spaces $G/K$ of projective special K\"ahler type with a suitable freely acting discrete subgroup $\Gamma\subset G$, as follows: Recall first that there is only a finite list of non-compact globally symmetric projective special K\"ahler manifolds $G/K$~\cite{Cremmer:1984hc}. A lift of the $G$-action on $G/K$ to the (trivial) symplectic bundle $H\to G/K$ (see Section~\ref{sec:special-geometry}) is specified by a homomorphism\footnote{Here we have fixed a K\"ahler frame to trivialise the line bundle $L\to G/K$ as discussed in Section~\ref{sec:special-geometry}.}
\begin{equation}
    \rho:G\to{\rm Sp}(2n+2,\mathbb{R}),\,g\mapsto M(g)\,
\end{equation}
satisfying the usual cocycle conditions. For a freely acting discrete subgroup $\Gamma\subset G$, we divide both $G/K$ and the bundle $H\to G/K$ by the action of $\rho\vert_\Gamma$ to find a projective special K\"ahler structure on the locally symmetric quotient $M_n=\Gamma\setminus G/K$. 
The electric-magnetic duality group in this case is given by $\rho(\Gamma)\subset{\rm Sp}(2n+2,\mathbb{R})$. 

For such locally symmetric projective special K\"ahler manifolds $M_n$, we first note that Quantum Gravity Principle~\ref{Swamp1} holds if and only if the representation $\rho$ maps $\Gamma$ into a subgroup of ${\rm Sp}(2n+2,\mathbb{Z})$. Furthermore, Quantum Gravity Principle~\ref{Swamp2} is satisfied if $\Gamma$ is an arithmetic subgroup of $G$.\footnote{By Margulis' arithmeticity theorem, any irreducible lattice $\Gamma$ for semisimple $G$ with ${\rm rk}(G)\geq2$ is arithmetic.} Indeed, in this situation the compactification theorem of Bailey--Borel states that $M_n=\Gamma\setminus G/K$ is quasi-projective and as such a projective simple normal crossing compactification $\overline{M}_n$ exists. Consequently, all multi-parameter infinite distance limits in $\overline{M}_n$ are of simple normal crossing type and the full analysis of Section~\ref{sec:Asymptotics} applies. In particular, this means that in all single-parameter as well as EFT string-like multi-parameter limits~\eqref{eq:homo-scaling} the physical 1- and 2-form couplings follow the rates~\eqref{eq:distance-scaling}. Notice also that the quasi-projectivity of $M_n$ (equivalent to the algebraic compactifiability condition of~\cite{Delgado:2024skw}) implies the geometric compactifiability of~\cite{Delgado:2024skw}, i.e. the volume of geodesic balls in $M_n$ grows no faster than those in the same dimensional affine space.

Locally symmetric moduli spaces $\Gamma\setminus G/K$ with $\Gamma$ arithmetic have been studied in much greater generality, in particular without any supersymmetry (or projective special K\"ahler) assumptions, in the recent works~\cite{Baines:2025upi,Baines:2026aug}. Interestingly, these references show that the Emergent String Conjecture decay rates~\eqref{eq:distance-scaling} in geodesic infinite distance limits follow also from the group-theoretic structure of locally symmetric spaces.\footnote{Geodesics in these spaces are of the form $\gamma(t)=g\exp(tX)\cdot o$ in the notation of~\cite{Baines:2025upi,Baines:2026aug}. As such, geodesic infinite distance limits are of the EFT string limit form~\eqref{eq:homo-scaling}.} In particular, in case the locally symmetric moduli space is also a projective special K\"ahler manifold, this latter structure is not needed to infer the asymptotic decay rates~\eqref{eq:distance-scaling}. This further corroborates the point made in previous sections. Namely, the existence of the tower of states as well as its spacing and degeneracies constitutes the actual quantum gravitational content of the Emergent String Conjecture. 

One might also wonder about globally symmetric projective special K\"ahler manifolds $G/K$ of non-compact type, which have been classified in~\cite{Cremmer:1984hc}. These spaces are contractible, implying that the symplectic bundle $H\to G/K$ is trivial, $H\simeq G/K\times\mathbb{R}^{2n+2}$. Furthermore, as $\pi_1(G/K)=0$ there is no non-trivial monodromy on globally symmetric projective special K\"ahler manifolds. Seeing that $H$ is globally trivial, one can trivially find an integral structure $H_\mathbb{Z}$ underlying $H$, $H_{\mathbb{Z},t}=\{t\}\times\mathbb{Z}^{2n+2}$ for every $t\in G/K$. However, this structure is associated with an additional choice of generator of $\mathbb{Z}^{2n+2}\subset\mathbb{R}^{2n+2}$ and so constitutes extra data that is not given from the supergravity alone. Without the choice of such a generator, globally symmetric projective special K\"ahler manifolds therefore do not satisfy Quantum Gravity Principle~\ref{Swamp1}. While we cannot show that globally symmetric moduli spaces explicitly violate also Quantum Gravity Principle~\ref{Swamp2}, we can prove that they are never quasi-projective. The latter is (in practice) the most common sufficient condition for the existence of a simple normal crossing compactification. To see this, recall that if $G/K$ were quasi-projective, it would be biholomorphic to some quasi-projective variety $X$ for which there exists a projective compactification $\bar{X}$. It follows that every bounded holomorphic function on $X$ must in fact be constant.\footnote{This is an application of the removable singularity theorem as the local model of the simple normal crossing boundary divisor $\bar{X}\setminus X$ is a product of punctured disks.} However, by the Harish--Chandra embedding theorem $G/K$ is biholomorphic to a bounded domain in some complex affine space $\mathbb{C}^N$ for which all coordinate functions are bounded and holomorphic, but not constant. Hence, $G/K$ is not quasi-projective and in the above sense it is unlikely that $G/K$ satisfies Quantum Gravity Principle~\ref{Swamp2}.

\section{Positivity of cubic prepotential couplings}\label{sec:positive}
The results of the previous sections raise the question which further properties of supergravity theories that derive from Calabi--Yau compactifications of string theory are in fact a consequence of 4d ${\cal N}=2$ supersymmetry alone. 
As we now discuss, one such property is the existence of a special basis in which the (classical) Chern--Simons couplings are non-negative. 

To appreciate this point, recall first from Section~\ref{ssec:local} that the projective special K\"ahler geometry of $M_n$ dictates the prepotential $\mathbb{z}$ to be homogeneous of degree two. The classification of infinite distance limits in $M_n$ in terms of the associated log-monodromy matrix $N$ puts further constraints on the form of $\mathbb{z}$ near infinite distance boundaries. Combining the expressions~\eqref{eq:sg-Kpot} and~\eqref{eq:Hodge-Kpot} for the moduli space K\"ahler potential, it follows from $N^{d+1}=0$ for all $d\geq3$ that any rational piece of the prepotential $\mathbb{z}$ can at most have a cubic numerator in the special coordinates $Z^I$. We will focus on this cubic rational part in this section, i.e. we consider prepotentials of the form
\begin{equation}\label{eq:F-cubic}
    \mathbb{z}=-\frac{C_{ijk}}{6}\frac{Z^iZ^jZ^k} {Z^0}   \,,
\end{equation}
which arise most prominently near infinite distance boundaries of type IV$_{n}$ (i.e. points of maximal unipotent monodromy).

The fully symmetric couplings $C_{ijk}$ enter the effective action~\eqref{eq:4d-action} in the last term as
\begin{equation}
    S_{\rm 4d}\supset\int_{\mathbb{R}^{1,3}}\frac{C_{ijk}}{2}{\rm Re}(t^i)F^i\wedge F^j\,,
\end{equation}
which shows that they are quantised as $C_{ijk}\in\mathbb{Z}$. We will therefore refer to these couplings as (cubic) Chern--Simons couplings. The goal of this section is to establish the following\footnote{This statement can be extended to any type IV$_d$, $d<n$, boundary of $M_n$ which by itself is connected to some type IV$_n$ point. It would be interesting to investigate whether the absence of a type IV$_n$ point also implies the absence of limits of type IV$_d$ with $d<n$.}
\begin{claim}
    Consider a 4d ${\cal N}=2$ supergravity and a neighbourhood of a type IV$_n$ degeneration of the vector multiplet moduli space $M_n$, where the prepotential admits an expansion of the form \eqref{eq:F-cubic}. Then there exists a basis of vector multiplets $\{(\tilde{A}^i,\tilde{t}^i)\}$ in which
 $\tilde{C}_{ijk} \geq 0$.
\end{claim}
In supergravity theories obtained from Calabi-Yau compactifications, this basis is simply the geometric K\"ahler basis, but in general supergravity theories the existence of such a basis is a non-trivial statement which requires further justification.

\begin{proof}[Argument]
The fact that prepotentials of the form~\eqref{eq:F-cubic} arise near boundaries of type IV$_n$ is another indication that these limits correspond to decompactification limits to 5d, as mentioned already in Section~\ref{ssec:ESC}. The reason is that such prepotentials are known to lie in the image of the supergravity $r$-map~\cite{deWit:1991nm,Alekseevsky:2008cta,Cortes:2011aj}: $M_n$ can be written as the tangent bundle of a projective special real manifold $\cM_{n-1}\subset\mathbb{R}^{n}$~\cite{Gunaydin:1983bi}, i.e. $M_n=T(\mathbb{R}\cdot\cM_{n-1})\simeq\mathbb{R}^n\oplus\ii(\mathbb{R}\cdot\cM_{n-1})$. In physical terms this means that the real $n$-dimensional slice in $M_n$ where all axions vanish gives rise to 5d $\cN=1$ supergravity theory coupled to $n-1$ vector multiplets. The cubic prepotential of this 5d ``parent theory'' takes the simple form  
\begin{equation}\label{eq:F-5d}
    \cF(X)=\frac{1}{6}C_{ijk}X^iX^jX^k\,,
\end{equation}
and its vector multiplet moduli space is given by the hypersurface $\cM_{n-1}=\{\cF(X)=1\}\subset\mathbb{R}^{n}$. Circle reduction of this 5d theory, i.e. applying the $r$-map, gives back the original 4d supergravity specified by~\eqref{eq:F-cubic} upon identifying $X^i=e^{K/3}{\rm Im}(t^i)$ in the gauge $Z^0=1$, $Z^i=t^i$.\footnote{The reason for the factor $e^{K/3}$ in the identification of the vector multiplet scalars is the hypersurface constraint $\cF(X)=1$ of the 5d theory.} Such 5d supergravities are the topic of~\cite{Kaufmann:2024gqo,david-vicente}, where it is shown that there always exists a basis of 5d gauge fields (including the 5d graviphoton) $\tilde{A}_{\rm 5d}^i$, the so-called K\"ahler basis, in which $\tilde{C}_{ijk}\geq0$. The proof is based on two main ingredients~\cite{Kaufmann:2024gqo,david-vicente}. First, the signature of the Hessian matrix $(\cF_{ij})$ is fixed throughout $\overline{\cM}_{n-1}$ to be $(1,r)$ with $r\leq n-2$. Second, for $p\in\overline{\cM}_{n-1}$, the functional
\begin{equation} \label{5dtension}
    T_p:\cM_{n-1}\to\mathbb{R},\,X\mapsto T_p(X)=\frac{1}{6}C_{ijk}X^iX^jp^k
\end{equation}
can be shown to be non-negative~\cite{david-vicente}. These two properties are sufficient to prove that $\tilde{C}_{ijk} \geq 0$ in a suitable basis \cite{Kaufmann:2024gqo}.\footnote{Under an appropriate completeness assumption, $T_p$ can be interpreted as the tension of a BPS string with magnetic charge vector $p\in\cC_{\rm BPS}$~\cite{Katz:2020ewz,Kim:2024tdh}, in which case positivity of $T_p(X)$ follows from positivity of the BPS tension, and the signature of $(\cF_{ij})$ is also required by consistency of the anomalies on the worldsheet~\cite{Katz:2020ewz,Kim:2024tdh}. Importantly, however, just like the signature of $(\cF_{ij})$, positivity of $T_p(X)$ can be shown to be a general consequence of the projective special real property of $\cM_{n-1}$ alone \cite{david-vicente}, and hence of supersymmetry, without any additional assumption on the existence of BPS strings with tension $T_p(X)$.} 

The $r$-map/ circle reduction then makes a similar positivity result in 4d supergravities with prepotential~\eqref{eq:F-cubic} obvious. The change of basis to the K\"ahler basis in five dimensions induces a change of basis $\{(A^i,t^i)\}\to\{(\tilde{A}^i,\tilde{t}^i)\}$ in the 4d supergravity as depicted by the dashed horizontal line in the diagram
\begin{equation}
    \begin{tikzcd}
        \mathrm{SG}_5 \arrow[rr, "{\text{K\"ahler\text{-}basis}}"] &  &
        \widetilde{\mathrm{SG}}_5 \arrow[d, "r"] \\
        \mathrm{SG}_4 \arrow[u, "r^{-1}"] \arrow[rr, dashed]         &  & 
        \widetilde{\mathrm{SG}}_4            
        \end{tikzcd}\,.
\end{equation}
Denoting the 5d change of basis by ${\cal A}\in{\rm GL}_n(\mathbb{R})$, the 4d period vector transform as
\begin{equation}
    \Pi=\left(\begin{array}{c}
        Z^I \\ \mathbb{z}_I
    \end{array}\right)\mapsto\tilde{\Pi}=\left(\begin{array}{cc}
        e^{-K/3}{\cal A} & 0 \\
        0 & e^{K/3}({\cal A}^{-1})^T
    \end{array}\right)\left(\begin{array}{c}
        Z^I \\ \mathbb{z}_I
    \end{array}\right)=\left(\begin{array}{c}
        \tilde{Z}^I \\ \tilde{\mathbb{z}}_I
    \end{array}\right)\,,
\end{equation}
where it is crucial that the covector $\mathbb{z}_I$ is written as a column vector in $\Pi$. This shows that the induced 4d change of basis is a symplectic transformation in the 4d supergravity which, as discussed in Section~\ref{sec:special-geometry}, leaves the 4d physics invariant. We have therefore established the existence of a K\"ahler basis also in 4d $\cN=2$ supergravities with a prepotential of the form~\eqref{eq:F-cubic}. 
\end{proof}

It is natural to ask whether the existence of a 4d K\"ahler basis can be deduced without reference to a higher-dimensional theory. Repeating the same procedure as in 5d, our first idea would be to base such a proof on properties of ${\rm Im}(\mathbb{z}_I)$, the 4d analogue of \ref{5dtension}, whose non-negativity, however, is not obvious. Since we recover the geometry of the vector multiplet moduli space of the 5d supergravity by restricting to zero axions in 4d, we believe that the argument presented above is the most direct way of inferring the existence of a K\"ahler basis in 4d $\cN=2$ supergravity theories near type IV$_n$ boundaries.\footnote{This is in fact further substantiated by considering the special class of 4d $\cN=2$ vector multiplet moduli spaces arising as complexified K\"ahler moduli spaces in Calabi--Yau compactifications of Type IIA string theory. Indeed, the $C_{ijk}$ correspond to the classical triple intersection numbers of the underlying Calabi--Yau whose geometry is encoded purely in the saxionic parts of the vector multiplet scalars, whereas the axionic parts are governed by the $B$-field.}

\section{Discussion}\label{sec:discussion}
We have analysed the projective special K\"ahler geometry of vector multiplet moduli spaces of 4d $\cN=2$ supergravity theories from the perspective of the Distance and Emergent String Conjectures. As formulated in~\cite{Ferrara:1991np,Ceresole:1992su}, this geometry gives rise to a {\it real} variation of polarised Hodge structure (VPHS), the explicit construction of which we have presented in Section~\ref{sec:special-geometry}.\footnote{A different construction in terms of Legendre submanifolds for a certain contact structure is given in~\cite{Cecotti:1989kn}.} The theory of {\it integral} VPHS with quasi-unipotent monodromy has been crucial to recent developments concerning the Distance and Emergent String Conjectures~\cite{Grimm:2018ohb,Grimm:2018cpv,Corvilain:2018lgw,Grimm:2019wtx,Grimm:2019bey,Grimm:2020cda,Bastian:2020egp,Grimm:2021ikg,Bastian:2021eom,Bastian:2021hpc,Grimm:2021ckh,Hassfeld:2025uoy,Monnee:2025ynn,Monnee:2025msf} in Calabi--Yau compactifications of Type II string theory. For a general 4d $\cN=2$ supergravity, however, the existence of such an integral structure is not guaranteed and so the classic results of~\cite{schmid,CKS} are, a priori, not applicable. Interestingly, the very recent works~\cite{Deng2022OnTN,complexVHS,complexVHS-multi} generalise large parts of the results of~\cite{schmid,CKS}, in particular the nilpotent orbit theorem as well as the Hodge norm estimates, to complex VPHS without underlying integral structure. These more general mathematical results have allowed us to extend the classification and analysis of infinite distance limits, originally performed in~\cite{Grimm:2018ohb,Grimm:2018cpv,Corvilain:2018lgw} for Type IIB compactifications, to all 4d $\cN=2$ vector multiplet moduli spaces, irrespective of the existence of a UV completion. 

More precisely, the analysis of single-parameter infinite distance limits carries over without any assumptions to the general supergravity setting, thanks to the results of~\cite{complexVHS}. This means that the minimal physical 1- and 2-form couplings $\mathfrak{q}_{\rm min}$ and $\cQ_{\rm min}$ (defined in~\eqref{eq:1-form} and~\eqref{eq:2-form}) scale as specified in Claim \ref{Claim1}. This shows perfect agreement with the predictions of the Emergent String Conjecture and is a first hint that within the context of the vector multiplet moduli space of 4d ${\cal N}=2$ supergravity, the exponential rates given in~\eqref{eq:distance-scaling} do not constitute quantum gravitational data, but rather follow from supersymmetry alone. For multi-parameter limits, we have to restrict to those of simple normal crossing type, to which the mathematical analysis of~\cite{Deng2022OnTN,complexVHS-multi} applies. Within this class of limits, there is again perfect agreement between the behaviour of the minimal physical 1- and 2-form couplings and the Emergent String Conjecture. We find it quite remarkable that, at the level of gauge couplings, the supergravity is restrictive enough to only allow for infinite distance limits enforcing the dichotomy predicted by quantum gravity. 

These findings can be of practical use in situations where quantum corrections obscure the correct asymptotic behaviour in the vector multiplet moduli space. Even if the exact form of the corrections may not be known, the projective special K\"ahler geometry guarantees that the infinite distance loci in the fully corrected vector multiplet moduli space must respect the scalings expected from quantum gravity. For example, \cite{Kaufmann:2026fli,Kaufmann:2026mha} discussed the infinite distance behaviour in the vector multiplet moduli space of an orientifold of Type IIB string theory on ${\rm K3} \times T^2$. While at the classical level, the prepotential is of $STU$-type, perturbative and non-perturbative corrections \cite{Camara:2008zk,Berg:2005ja} obstruct infinite distance limits in the naive, classical coordinates. It is therefore welcome news that the projective special K\"ahler geometry and its VPHS guarantee that the asymptotic boundaries of the quantum corrected vector multiplet moduli space are in agreement with the expectations from the Emergent String Conjecture. While in the specific example, the asymptotic form of the quantum corrected flat coordinates is obvious already from duality to a Type IIA compactification on a Calabi-Yau threefold \cite{Antoniadis:1997eg,Kaufmann:2026fli,Kaufmann:2026mha}, similar conclusions can be drawn also in more complicated settings where a controlled dual formulation may not be available. It would be interesting to investigate this further.

While the asymptotic scaling of the 1- and 2-form couplings, at least in single-parameter limits, is  fully dictated by 4d ${\cal N}=2$ supergravity, there are important consistency conditions which do not follow from supersymmetry alone. We have identified in particular the existence of an underlying integral structure and the existence of a simple normal crossing compactification as such such additional conditions. In Calabi-Yau compactifications, both are automatic. More generally, we have linked these properties to universal quantum gravity principles. Regarding the underlying integral structure, such a connection was anticipated in~\cite{Cecotti:1989kn} and we have identified the relevant quantum data to be the completeness hypothesis as well as absence of global symmetries: Without an integral structure, mutually non-rational charges can exist and lead to a violation of both these principles \cite{Banks:2010zn}. This motivates our Quantum Gravity Principle~\ref{Swamp1}, which identifies an integral structure as a necessary condition for the supergravity theory to admit a UV completion. We have also found a physical interpretation of the simple normal crossing condition underlying the analysis of~\cite{Deng2022OnTN,complexVHS-multi}: Normal crossing is  necessary to ensure consistency with the Distance and Emergent String Conjectures in the sense that it guarantees well-defined duality frames emerging at the intersection of any two single-parameter limits. 
If normal crossing is violated, an inconsistency arises even if the theory has only a finite number of states, and hence this condition is independent of the appearance of any infinite tower. In Example~\ref{ex:non-snc} we have demonstrated that  normal crossing does not follow automatically from the projective special K\"ahler geometry, which led to the formulation of Quantum Gravity Principle~\ref{Swamp2}. Interestingly, a sufficient condition for simple normal crossing, quasi-projectiveness, also implies compactifiability of the moduli space in the sense of \cite{Delgado:2024skw}.

While both quantum gravity principles formulated in Section~\ref{sec:swamp} are strictly necessary in order for a 4d $\cN=2$ supergravity to have a UV completion within string theory, they are by no means sufficient. Even ignoring the hypermultiplet sector completely, what remains to be shown when these principles hold is the existence of an infinite tower of states becoming light at infinite distance in the vector multiplet moduli space. In other words, it is  the tower itself together with its properties (spacing, degeneracies) which constitutes the actual quantum gravitational content of the Emergent String Conjecture. In supergravity theories inherited from Calabi-Yau threefolds, the explicit geometry of the latter is key to establish the existence of these towers \cite{Lee:2019wij,Hassfeld:2025uoy,Monnee:2025ynn}. It would  be very important to find further compelling bottom-up (and probably indirect) arguments for their existence and properties, possibly along the lines of~\cite{Cribiori:2023ffn,Basile:2023blg,Basile:2024dqq,Herraez:2024kux,Bedroya:2024ubj}.

\vspace{0.5cm}

\noindent{\bf Acknowledgements.}
We thank Vicente Cort\'es, Damian van de Heisteeg, David Lindemann, Jeroen Monnee and Max Wiesner for useful discussions. We are particularly grateful to Vicente Cort\'es and Jeroen Monnee for pointing out several important references to us, and to Christian Schnell for sharing the upcoming work \cite{complexVHS-multi} with us. We also thank Stefano Lanza for collaboration during the early stages of this work. This work is supported in part by Deutsche Forschungsgemeinschaft under Germany’s Excellence Strategy EXC 2121 Quantum Universe 390833306, by Deutsche Forschungsgemeinschaft through a German-Israeli Project Cooperation (DIP) grant “Holography and the Swampland” and by Deutsche Forschungsgemeinschaft through the Collaborative Research Center 1624 “Higher Structures, Moduli Spaces and Integrability”. 

\bibliographystyle{jhep}
\bibliography{references.bib}

\end{document}